\documentclass[a4paper,onecolumn,11pt,accepted=2022-06-20]{quantumarticle}
\pdfoutput=1
\usepackage{stmaryrd}
\usepackage[english]{babel}
\usepackage[autostyle, english = american]{csquotes}
\usepackage{authblk}

\usepackage{tikz}

\usetikzlibrary{
  arrows,
  automata,
  shapes,
  shapes.geometric,
  positioning,
  calc,
  chains,
  decorations.pathmorphing,
  intersections}

\usepackage{hyperref}
\hypersetup{pdfpagemode=UseNone}

\usepackage{amssymb, amsthm, amsmath, dsfont, mathtools}
\usepackage{latexsym}
\usepackage{eepic}
\usepackage{framed}
\mathtoolsset{showonlyrefs}
\usepackage{blkarray, bigstrut}
\usepackage[title]{appendix}
\usepackage[numbers,sort&compress]{natbib}

\newtheorem{theorem}{Theorem}
\newtheorem{lemma}[theorem]{Lemma}

\newtheorem{cor}[theorem]{Corollary}

\newtheorem{fact}[theorem]{Fact}
\newtheorem{claim}[theorem]{Claim}
\theoremstyle{definition}
\newtheorem{definition}[theorem]{Definition}

\newtheorem{example}[theorem]{Example}

\newcommand{\tinyspace}{\mspace{1mu}}

\newcommand{\tr}{\operatorname{Tr}}
\newcommand{\rank}{\operatorname{rank}}
\renewcommand{\det}{\operatorname{Det}}

\renewcommand{\t}{{\scriptscriptstyle\mathsf{T}}}

\newcommand{\abs}[1]{\lvert #1 \rvert}

\newcommand{\ip}[2]{\langle #1 , #2\rangle}

\newcommand{\floor}[1]{\lfloor #1 \rfloor}

\newcommand{\norm}[1]{\lVert\tinyspace #1 \tinyspace\rVert}

\newcommand{\ket}[1]{\lvert\tinyspace #1 \tinyspace \rangle}

\renewcommand{\t}{{\scriptscriptstyle\mathsf{T}}}

\newcommand{\I}{\mathds{1}}

\newcommand{\setft}[1]{\mathrm{#1}}

\newcommand{\Pos}{\setft{Pos}}

\newcommand{\Lin}{\setft{L}}

\newcommand{\density}[1]{\setft{D}\left(#1\right)}

\newcommand{\herm}[1]{\setft{Herm}\left(#1\right)}
\newcommand{\pos}[1]{\setft{Pos}\left(#1\right)}

\newcommand{\complex}{\mathbb{C}}

\newcommand{\proj}{\mathbb{P}}

\newenvironment{namedtheorem}[1]
	       {\begin{trivlist}\item {\bf #1.}\em}{\end{trivlist}}

\newcommand\W{\mathcal{W}}

\newcommand\V{\mathcal{V}}

\renewcommand{\Im}{\setft{Im}}

\newcommand{\lip}{\langle}
\newcommand{\rip}{\rangle}
\newcommand{\braket}[2]{\lip #1 \vert #2 \rip}

\DeclareMathOperator{\spn}{\setft{span}}
\DeclareMathOperator{\Ima}{\setft{Im}}

\DeclareMathOperator{\Gr}{\setft{Gr}}
\DeclareMathOperator{\GL}{\setft{GL}}

\DeclareMathOperator{\Un}{\setft{U}}
\DeclareMathOperator{\Sp}{\setft{S}}
\newcommand\seg{\setft{Seg}}

\newcommand\sgn{\setft{sgn}}

\makeatletter
\def\Ddots{\mathinner{\mkern1mu\raise\p@
		\vbox{\kern7\p@\hbox{.}}\mkern2mu
		\raise4\p@\hbox{.}\mkern2mu\raise7\p@\hbox{.}\mkern1mu}}
\makeatother

\begin{document}
\emergencystretch 3em

\title{Entangled subspaces and generic local state discrimination with pre-shared entanglement}
\author{Benjamin Lovitz}
\email{benjamin.lovitz@gmail.com}
\affiliation{Institute for Quantum Computing and Department of Applied Mathematics, University of Waterloo, 200 University Ave W, Waterloo, ON, Canada}
\orcid{0000-0002-2445-2701}
\author{Nathaniel Johnston}
\email{njohnston@mta.ca}
\affiliation{Department of Mathematics \& Computer Science, Mount Allison University, Canada}
\affiliation{Department of Mathematics \& Statistics, University of Guelph, Canada}
\maketitle
\begin{abstract}
Walgate and Scott have determined the maximum number of generic pure quantum states that can be unambiguously discriminated by an LOCC measurement~\cite{article}. In this work, we determine this number in a more general setting in which the local parties have access to pre-shared entanglement in the form of a resource state. We find that, for an arbitrary pure resource state, this number is equal to the Krull dimension of (the closure of) the set of pure states obtainable from the resource state by SLOCC. Surprisingly, a generic resource state maximizes this number.

Local state discrimination is closely related to the topic of entangled subspaces, which we study in its own right. We introduce $r$-entangled subspaces, which naturally generalize previously studied spaces to higher multipartite entanglement. We use algebraic-geometric methods to determine the maximum dimension of an $r$-entangled subspace, and present novel explicit constructions of such spaces. We obtain similar results for symmetric and antisymmetric $r$-entangled subspaces, which correspond to entangled subspaces of bosonic and fermionic systems, respectively.
\end{abstract}

\section{Introduction}

An \textit{LOCC measurement} is a quantum measurement that can be implemented by local operations and classical communication (LOCC). We say an $n$-tuple of pure quantum states $([v_1],\dots, [v_n])$ is \textit{locally (unambiguously) discriminable} if there exists an LOCC measurement with $n+1$ outcomes $\{1,\dots, n, ?\}$ that, when performed on any $[v_a]$, outputs either $a$ or $?$, with non-zero probability to output $a$. More generally, we say that $([v_1],\dots, [v_n])$ is \textit{locally (unambiguously) discriminable with (a pure resource state)} $[w]$ if $([v_1],\dots, [v_n])$ can be locally discriminated using an LOCC measurement with pre-shared entanglement $[w]$. Our notation for pure states $[w]$ is projective, and can be thought of as $\frac{1}{\norm{w}^2}ww^*$ for a non-zero vector $w$, where $\norm{\cdot}$ denotes the Euclidean norm and $(\cdot)^*$ denotes the conjugate-transpose.

For an algebraic variety $Y,$ we say that a property holds for a~\textit{generic} element of $Y$ if there exists a Zariski-open-dense subset of $Y$ on which that property holds. Zariski-open-dense sets are massive. In particular, they are full-measure with respect to any absolutely continuous measure, e.g. the unitary Haar measure. One should intuitively think of a property holding for a generic element of $Y$ if it holds with probability one when an element of $Y$ is picked at random.

Walgate and Scott determined that, for $m$ local spaces of (affine) dimensions $d_1,\dots, d_m$, a generic $n$-tuple of ${\sum_{j=1}^m (d_j-1)+1}$ pure states is locally discriminable, and this is the largest number for which this holds \cite{article}. We extend Walgate and Scott's result to determine the number of generic pure states that can be locally discriminated with an arbitrary pure resource state $[w]$. For example, we prove that at most ${r(\sum_{j=1}^m (d_j-1)+1)}$ generic pure states are locally discriminable with the tensor-rank-$r$ GHZ state ${[\tau_{r,m}]=\left[\sum_{a=1}^r e_a^{\otimes m}\right]}$, where $e_1,\dots, e_r$ are standard basis vectors, and this bound is often achieved (Corollary~\ref{cor:GHZ}). For $m=2$, this bound cannot be achieved: precisely ${d_1d_2- (d_1-\min\{d_1,r\})(d_2-\min\{d_2,r\})}$ generic pure states are locally discriminable with $[\tau_{r,2}]$ (or any other Schmidt rank $r$ state), and this is the largest number for which this holds (Corollary~\ref{cor:bipartite}).

More generally, we characterize this number for an arbitrary pure resource state $[w]$ (Theorem~\ref{USD_characterization}). The \textit{SLOCC image} of $[w]$, denoted by $\Ima([w])$, is the set of pure states obtainable from $[w]$ by stochastic local operations and classical communication (SLOCC). Letting ${d=\dim(\overline{\Ima([w])})}$ be the Krull dimension of the Zariski closure of $\Ima([w])$, we prove that a generic $(d+1)$-tuple of pure states is locally discriminable with $[w]$. We also prove a ``strong converse:" A Zariski-open-dense set of $(d+2)$-tuples of pure states are \textit{not} locally discriminable with $[w]$. To recover Walgate and Scott's result, observe that the SLOCC image of a trivial resource state is the set of unentangled (or, \textit{product}) pure states. This set is already closed, and has dimension $d=\sum_{j=1}^m (d_j-1)$.

Our characterization extends Walgate and Scott's result in two directions beyond the obvious addition of a resource state: First, our results hold under the algebraic-geometric notion of ``generic" introduced above, which yields stronger statements than the measure-theoretic notion used by Walgate and Scott. Second, our ``strong converse," mentioned above, is much stronger than the converse statement proven by Walgate and Scott (see the discussion following Corollary~\ref{WS_generic}).

It is natural to quantify ``how useful" a given resource state $[w]$ is for LUSD in terms of the number of generic pure states that can be locally discriminated with $[w]$. We use algebraic group theory to prove that, under this barometer, a generic resource state is maximally useful for LUSD. This is quite surprising, because in many other contexts the most useful quantum states form a measure zero subset of the Hilbert space. For example, in bipartite space the maximally entangled states form a measure zero set, and in multipartite space the set of graph states (a class of states often regarded as the most useful states) has measure zero~\cite{PhysRevA.69.062311}.
As one more example, it is known that most multipartite states are, in a sense, useless for measurement-based quantum computation \cite{PhysRevLett.102.190501}.

We now sketch a proof of one direction (the ``strong converse") of our characterization of generic LUSD with a resource state. We invoke a recent observation of Bandyopadhyay et al., that an $n$-tuple of pure states $([v_1],\dots, [v_n])$ is locally discriminable with $[w]$ if and only if there exist pure states $[u_1],\dots, [u_n]\in \Ima([w])$ for which $({u_a^\t v_b \neq 0 \iff a=b})$ \cite{PhysRevA.94.022311}. In particular, ${\Ima([w]) \cap \spn\{[v_1],\dots, [v_{n-1}]\}^\perp\neq \{\}}$, because the intersection contains $[u_n]$. We also use a theorem in algebraic geometry, which states that for a quasiprojective variety $X$ of dimension $d$, a generic (projective) linear subspace of codimension $d+1$ is disjoint from $X$, and a generic linear subspace of smaller codimension intersects $X$. Letting $X=\overline{\Ima([w])}$, it follows from the algebraic geometry result that for a generic $(d+1)$-tuple of states $([v_1],\dots, [v_{d+1}])$, it holds that ${X \cap \spn\{[v_1],\dots, [v_{d+1}]\}^\perp= \{\}},$ since $\spn\{[v_1],\dots, [v_{d+1}]\}^\perp$ forms a generic projective linear subspace of codimension $d+1$.  By the result of Bandyopadhyay et al., $([v_1],\dots, [v_{d+1}],[v_{d+2}])$ is not locally discriminable with $[w]$ for any pure state $[v_{d+2}]$. This proves that a generic $(d+2)$-tuple of pure states is not locally discriminable with $[w]$. We prove the other direction, that a generic $(d+1)$-tuple of pure states is locally discriminable with $[w]$, using a similar, but more complicated argument.

\begin{center}$\ast$~$\ast$~$\ast$\end{center}

We use similar algebraic-geometric techniques to study entangled subspaces. The \textit{tensor rank} of a pure state is the minimum number $r$ for which that state can be written as a superposition of $r$ product states. It is known that the set of pure states of tensor rank at most $r$ is precisely $\Ima([\tau_{r,m}])$. We define an \textit{$r$-entangled subspace} to be a projective linear subspace that avoids $\overline{\Ima([\tau_{r,m}])}$ (this closure is known as the set of pure states of \textit{border rank} at most $r$). In bipartite space, the tensor rank is equal to the Schmidt rank, and $\Ima([\tau_{r,2}])$ is already closed. Parthasarathy determined the maximum dimension of a 1-entangled subspace \cite{Parthasarathy:2004aa}, and Bhat explicitly constructed a 1-entangled subspace of maximum dimension~\cite{BHAT_2006}. Cubitt et al. proved analogous results for $r$-entangled subspaces of bipartite space \cite{Cubitt_2008}. A theorem in algebraic geometry states that if $X$ is a projective variety, then the minimum codimension of a projective linear subspace disjoint from $X$ is $\dim(X)+1$ (and, by the algebraic geometry result of the previous paragraph, almost all subspaces of this codimension avoid $X$). It follows, from a standard upper bound on $\dim(\overline{\Ima([\tau_{r,m}])}),$ that there always exists an $r$-entangled subspace of dimension
\begin{align}
d_1 \cdots d_m - r\sum_{j=1}^m (d_j-1)-r-1,
\end{align}
whenever this quantity is non-negative. Furthermore, this is often the maximum dimension of such a subspace. 
Using this bound, we explicitly construct $2$-entangled subspaces of maximum dimension in tripartite space with local (affine) dimensions $d_1, d_2 \in \{2,3\}$ and $d_3=2$ (i.e. qubits and qutrits); and in quadripartite space with local (affine) dimensions $d_1=d_2=d_3=d_4=2$ (i.e. all qubits). To show that these subspaces are indeed 2-entangled, we solve an equivalent ideal membership problem using the Macaulay2 software package \cite{M2,entangled_subspaces_code}. While ideal membership problems are notoriously intractable in general, our positive results reveal that this may not be the case for verifying $r$-entangled subspaces (at least for small $r$).

We also define $r$-entangled subspaces of the symmetric and antisymmetric spaces, which correspond to bosonic and fermionic entangled subspaces, respectively \cite{Grabowski_2012}. We explicitly construct maximal symmetric and antisymmetric $r$-entangled subspaces of bipartite space for arbitrary $r$, and of multipartite space for $r=1$, which matches the cases of standard $r$-entangled subspace constructions presented in \cite{BHAT_2006,Cubitt_2008}.



It is known that, under various notions of entanglement, the maximum dimension of an entangled subspace is precisely the maximum number of negative eigenvalues of an entanglement witness \cite{Augusiak_2011, Johnston_2013,Johnston_2019}. The number of negative eigenvalues quantifies ``how good" the witness is at detecting entanglement. We prove that this connection between subspaces and negative eigenvalues holds under a much more general notion of witness, including multipartite $r$-entanglement witnesses.

There are other types of entangled subspaces that have been studied in previous works: \textit{non-positive partial transpose subspaces}, for which every mixed state supported on that subspace has non-positive partial transpose \cite{Johnston_2013, Johnston_2019}; \textit{genuinely entangled subspaces}, for which every element is non-product with respect to every bipartition \cite{Cubitt_2008,PhysRevA.98.012313}; and subspaces of bipartite space with high entropy of entanglement \cite{Hayden_2006}. Entangled subspaces are connected to unextendible product bases, and have found applications, for example, in quantum error correction \cite{PhysRevA.76.042309,PhysRevA.69.052330} and quantum tomography \cite{Heinosaari_2013}.
\begin{center}$\ast$~$\ast$~$\ast$\end{center}

This work is organized as follows. In Section~\ref{mp} we review some mathematical preliminaries for this work, including projective notation; the Segre, Veronese, and Grassmannian varieties; relevant results in algebraic geometry; the SLOCC image; and LUSD. In Section~\ref{sec:genericLUSD} we state and prove our characterization of generic LUSD with a resource state. In Section~\ref{sec:entangled_subspaces} we present explicit constructions of $r$-entangled subspaces of maximum dimension. In Section~\ref{sec:witnesses} we use our results to study entanglement witnesses. In Appendix~\ref{app:proofs} we prove several facts introduced in Section~\ref{mp}. In Appendix~\ref{app:dim} we write down the maximum dimensions of entangled subspaces, invoking the algebraic geometry result mentioned above and known results on the dimensions of secant varieties. In Appendix~\ref{app:entangled} we prove that a certain linear subspace introduced in Section~\ref{sec:entangled_subspaces} is entangled.

\section*{Acknowledgments}
We thank William Slofstra for his help on numerous occasions with the geometric arguments in this work. We thank Debbie Leung, Micha\l{} Oszmaniec and John Watrous for helpful discussions. B.L. was supported by the University of Waterloo and the Government of Ontario through an Ontario Graduate Scholarship. N.J. was supported by NSERC Discovery Grant number RGPIN-2016-04003.

\section{Preliminaries}\label{mp}
In this section, we review several necessary preliminaries for this work. In this preamble, we review some basic objects in affine and projective space, including the crucial observation that pure states can be identified with elements of projective space. In Section~\ref{sec:notation_qi_remarks} we introduce our notation, and discuss why it might be slightly unfamiliar to those with a quantum information background. In Section~\ref{sec:secant} we review the Segre, Veronese, and Grassmannian varieties, which correspond to the sets of unentangled pure states in standard, symmetric, and antisymmetric space, respectively. We also define the $r$-th secants to these varieties, which we will use to define $r$-entangled subspaces in these three settings. In Section~\ref{sec:n-1plane} we describe a useful alternative interpretation of the Grassmannian as the set of projective planes of a fixed dimension, and review some necessary preliminaries for our characterization of generic LUSD with a resource state. In Section~\ref{sec:alg_geo_result} we review a theorem in algebraic geometry on the dimension of a linear subspace disjoint from a variety. In Section~\ref{sec:LOCC} we review the SLOCC image. In Section~\ref{sec:LUSD} we review LUSD.

\subsection{Notation and Terminology}\label{sec:notation_qi_remarks}

For a non-zero vector space $\V$ (which we always take to be over $\complex$), we use $\braket{\cdot}{\cdot}: \V \times \V \rightarrow \complex$ to denote the standard Hermitian inner product given by ${\braket{u}{v}=u^* v}$, where $(\cdot)^*$ denotes the conjugate-transpose. We use $\Sp(\V) \subseteq \V$ to denote the unit sphere of $\V$ with respect to the norm induced by the inner product $\braket{\cdot}{\cdot}$.

We use $\proj (\V)$ (or more briefly, $\proj \V$) to denote the set of $1$-dimensional linear subspaces of $\V$, and we write ${[v] \in \proj \V}$ for the span of a non-zero vector $v \in \V$. While this projective notation is perhaps a bit non-standard in the field of quantum information theory, it is particularly important for us here: our work relies \emph{heavily} on results from algebraic geometry, where this notation is standard. Readers who are unfamiliar with projective geometry can still interpret our results and proofs as long as they keep a few points in mind:

\begin{itemize}
	\item In quantum physics, a ``pure state'' of $\V$ is typically represented by a unit vector $\ket{v} \in \Sp(\V)$, with the understanding that $\ket{v}$ and $\ket{w}$ represent the same state if there exists $\theta \in \mathbb{R}$ such that $\ket{w} = e^{i\theta}\ket{v}$. In this paper, we instead say that the pure states of $\V$ are represented by the $1$-dimensional linear subspaces $[v] \in \proj \V$, which is equivalent via the bijection which identifies a 1-dimensional linear subspace with some representative unit vector in that subspace.
	
	Identifying pure quantum states in this way as subspaces, rather than as unit vectors, is useful since it allows us to regard certain sets of pure quantum states as algebraic varieties, and study them using the machinery of algebraic geometry. However, not much will be lost by a reader making the mental substitution $[v] \rightarrow \ket{v}$ when reading the main results of this paper.
	
	For brevity, we will refer to pure quantum states simply as ``states,'' and we will refer to mixed quantum states (defined and treated only in Section~\ref{sec:witnesses}) as ``mixed states.''
	
	\item Instead of considering a linear subspace of vectors $Z \subseteq \V$, we typically consider the corresponding \emph{projective} linear subspace $\check{Z} \subseteq \proj \V$ defined by
	\begin{align}\label{eq:projectivize}
		\check{Z} = \big\{ [z] \in \proj \V : z \in Z \setminus \{0\} \big\}.
	\end{align}
	
	By convention, a projective linear subspace has dimension one less than the corresponding non-projective subspace (after all, we regard a $1$-dimensional line $\spn\{v\} \subseteq \V$ as a single point $[v] \in \proj \V$). For this reason, a \textit{projective $(n-1)$-plane} (or alternatively, a \textit{projective linear subspace of dimension $n-1$}) in $\proj \V$ is the projectivization of an $n$-dimensional linear subspace of $\V$. Similarly, we use the shorthand $\proj^{D}=\proj(\complex^{D+1})$, with the understanding that this is a projective $D$-dimensional space.
	
	As a result of considering projective subspaces instead of non-projective ones, the numbers that we present as the maximum dimension of a (projective) entangled subspace will be one less than the numbers presented in previous works (e.g., \cite{Parthasarathy:2004aa,Cubitt_2008}). 
\end{itemize}

We can also projectivize a set $Z \subseteq \V$ even if it is just an affine cone (i.e., closed under multiplication by complex scalars), rather than a linear subspace. In this case, the projectivization $\check{Z}$ is still given by Equation~\eqref{eq:projectivize}, but we note that it will no longer necessarily be a (projective) subspace---it will just be a set, potentially with no additional structure. Conversely, given a subset $X \subseteq \proj \V$, we define $\hat{X} \subseteq \V$ to be the affine cone over $X$:
\[
	\hat{X}=\big\{x \in \V : [x] \in X\big\}\cup \{0\}.
\]

For vector spaces $\W$ and $\V$, we use $\Lin(\W,\V)$ to denote the space of linear maps from $\W$ to $\V$, and let $\Lin(\V)=\Lin(\V,\V)$. If $m$ is a positive integer and we let $[m]=\{1,\dots, m\}$, then every vector space $\V=\bigotimes_{j=1}^m \V_j$ is canonically isomorphic to $\Lin(\bigotimes_{j \in S} \V_j^*, \bigotimes_{j \in [m]\setminus S} \V_j)$ for any subset $S \subseteq [m]$. The \textit{flattening rank} of a state $[v] \in \proj\V$ is the maximum, taken over all subsets $S \subseteq [m]$ of size $1 \leq \abs{S} \leq m-1$, of the matrix rank of $[v]$ as an element of (the projectivization of) $\Lin(\bigotimes_{j \in S} \V_j^*, \bigotimes_{j \in [m]\setminus S} \V_j)$.

Let $\Pos(\V)\subseteq \Lin(\V)$ be the set of positive semidefinite operators on $\V$, let $\Un(\V)\subseteq \Lin(\V)$ be the set of unitary operators on $\V$, and let $\GL(\V)\subseteq \Lin(\V)$ be the set of invertible operators on $\V$. Let $\{e_1,\dots, e_d\}$ be the standard basis of $\complex^d$.

\subsection{The Segre, Veronese, and Grassmannian varieties (and their secants)}\label{sec:secant}
We will be particularly interested in the algebraic varieties known as the Segre, Veronese, and Grassmannian varieties, as these correspond to the sets of unentangled states in a space of distinguishable, bosonic, and fermionic particles, respectively \cite{Grabowski_2012}. In this section, we briefly describe these varieties, as well as their $r$-th secants, which we will use to study $r$-entangled subspaces. We refer the reader to \cite{10.1112/S0024610706022630,landsberg2012tensors,harris2013algebraic,Bernardi_2018} for more in-depth treatments of these objects.

The Segre variety
\begin{align}\label{seg}
Y=\setft{Seg}(\proj^{d_1-1}\times \dots\times \proj^{d_m-1})\subseteq \textstyle{\proj (\bigotimes_{j=1}^m \complex^{d_j})}
\end{align}
is the image of the Segre embedding, and is equal to the set of states of the form ${[x_1 \otimes \dots \otimes x_m]}$, where $x_j \in \complex^{d_j}$ for each $j \in [m]$. The Segre variety corresponds to the unentangled (or, \textit{product}) states in the space $\proj(\bigotimes_{j=1}^m \complex^{d_j})$ of distinguishable particles.

For each permutation $\sigma \in S_m$, let $P_{\sigma} \in \Lin(\bigotimes^m \complex^d)$ be the linear map defined on product vectors as
\begin{align}
P_{\sigma}(x_1 \otimes\dots\otimes x_m)=x_{{\sigma^{-1}(1)}}\otimes \dots \otimes x_{{\sigma^{-1}(m)}},
\end{align}
and extended linearly. Let $\proj(\bigvee^m \complex^d)\subseteq \proj(\bigotimes^m \complex^d)$ denote the symmetric subspace, i.e, the set of states ${[x]\in \proj (\bigotimes^m \complex^d)}$ such that $P_{\sigma} x=x$ for all $\sigma \in S_m$. This space is spanned by states of the form $[x_1 \vee \dots \vee x_m]$, where
\begin{align}
[x_1 \vee \dots \vee x_m]=\left[ \sum_{\sigma \in S_m} x_{\sigma(1)}\otimes \dots \otimes x_{\sigma(m)}\right].
\end{align}
In quantum physics, $\proj(\bigvee^m \complex^d)$ represents a bosonic space of indistinguishable particles. (This space can also be thought of as the set of homogeneous polynomials of degree $m$ in $d$ variables.) The Veronese variety
\begin{align}
\nu_m(\proj^{d-1})\subseteq \textstyle{\proj({\bigvee}^m \complex^d)}
\end{align}
is the image of the $m$-th Veronese embedding, and is equal to the set of states in $\proj(\bigvee^m \complex^d)$ of the form $[x^{\vee m}]$. The Veronese variety corresponds to the set of unentangled states in the bosonic space $\proj({\bigvee}^m \complex^d)$.
%

Let $\proj(\bigwedge^m \complex^d)\subseteq \proj(\bigotimes^m \complex^d)$ denote the antisymmetric subspace, i.e. the set of states $[x]\in \proj (\bigotimes^m \complex^d)$ such that $P_{\sigma} x= (-1)^{\setft{sgn} (\sigma)} x$ for all $\sigma \in S_m$. This space is spanned by the set of states of the form $[x_1 \wedge \dots \wedge x_m]$, where
\begin{align}
[x_1 \wedge \dots \wedge x_m]=\left[ \sum_{\sigma \in S_m} (-1)^{\sgn(\sigma)} x_{\sigma(1)}\otimes \dots \otimes x_{\sigma(m)}\right].
\end{align}
In quantum physics, the antisymmetric subspace represents a fermionic space of indistinguishable particles. The Grassmannian variety
\begin{align}
\Gr(m-1,\proj^{d-1}) \subseteq \textstyle{\proj({\bigwedge}^m \complex^d)}
\end{align}
is the set of states in $\proj(\bigwedge^m \complex^d)$ of the form $[x_1 \wedge \dots \wedge x_m]$. The Grassmannian variety corresponds to the set of unentangled states in the fermionic space $\proj({\bigwedge}^m \complex^d)$. The Grassmannian can also be viewed as the variety of projective $(m-1)$-planes in $\proj^{d-1}$. We expound on this persective in Section~\ref{sec:n-1plane}, as we will make frequent use of it.

Now we introduce secant varieties. For a projective variety $X \subseteq \proj^D$, let
\begin{align}\label{sigma_r}
\sigma_r(X)=\overline{\bigcup_{[x_1],\dots, [x_r] \in X} \spn\{[x_1],\dots, [x_r]\}} \subseteq \proj^D
\end{align}
be the \textit{r-th secant variety} to $X$, where the closure can equivalently be taken with respect to either the Zariski or Euclidean topology. 
It is a standard result that
\begin{align}
\dim(\sigma_r(X)) \leq \min\{D, r\dim(X)+r-1\}.
\end{align}
If equality holds in this expression, then $\sigma_r(X)$ is said to have the \textit{expected dimension}, and otherwise it is said to be \textit{defective}.

For the Segre variety $Y\subseteq \proj(\bigotimes_{j=1}^m \complex^{d_j})$, we have $\dim(Y)=\sum_{j=1}^m (d_j-1)$ and $\dim(\proj(\bigotimes_{j=1}^m \complex^{d_j}))=d_1\cdots d_m-1$, so
\begin{align}\label{eq:dim_segre_secant}
\dim(\sigma_r(Y))\leq \min\bigg\{d_1\cdots d_m-1, r \sum_{j=1}^m (d_j-1)+r-1\bigg\}.
\end{align}
The \textit{border rank} of a state $[v] \in \proj(\bigotimes_{j=1}^m \complex^{d_j})$ is the smallest positive integer $r$ for which $[v] \in \sigma_r(Y)$. The \textit{tensor rank} of $[v]$ is the smallest $r$ for which $[v]$ is in the span of $r$ elements of $Y$ (i.e. $[v]$ is contained in the pre-closure of the set defined in~\eqref{sigma_r}, with $X=Y$). In the case of a bipartite quantum system (i.e., when $m = 2$), the border rank and tensor rank coincide, and are often referred to as the \textit{Schmidt rank}.

 A conjecturally complete set of defective $\sigma_r(Y)$ have been proposed in \cite{Abo_2008}, which are nicely summarized in \cite[Conjecture 6]{Bernardi_2018}. For example,
\begin{align}\label{eq:bipartite_sr_dim}
\dim(\sigma_r( \seg(\proj^{d_1-1}\times \proj^{d_2-1})))=d_1 d_2-(d_1-\min\{d_1,r\})(d_2-\min\{d_2,r\})-1,
\end{align}
so $\sigma_r( \seg(\proj^{d_1-1}\times \proj^{d_2-1}))$ is defective in many cases. Under the identification ${\complex^{d_1}\otimes \complex^{d_2} \cong \Lin((\complex^{d_1})^*,\complex^{d_2})}$, the variety $\sigma_r( \seg(\proj^{d_1-1}\times \proj^{d_2-1}))$ corresponds to the set of (projective) $d_2 \times d_1$ matrices of rank at most $r$.

For the Veronese variety $\nu_m(\proj^{d-1})\subseteq \proj({\bigvee}^m(\complex^d))$, we have $\dim(\nu_m(\proj^{d-1}))=d-1$ and $\dim(\proj({\bigvee}^m(\complex^d)))=\binom{d+m-1}{m-1}-1$, so
\begin{align}
\dim(\sigma_r(\nu_m(\proj^{d-1})))\leq \min\left\{\binom{d-1+m}{m}-1,rd-1\right\}.
\end{align}
A complete set of defective Veronese secants are known \cite{alexander1995polynomial}; see also Theorem 2 in \cite{Bernardi_2018}. In particular, $\sigma_r(\nu_2(\proj^{d-1}))$ is defective whenever $2 \leq r \leq d-1$, in which case
\begin{align}
\dim(\sigma_r(\nu_2(\proj^{d-1})))=\min\left\{\binom{d+1}{2}-1,rd-\binom{r}{2}-1\right\}.
\end{align}
Note that there is a typo in the expression of this dimension in Theorem 2 of \cite{Bernardi_2018}. Under the identification $\complex^{d}\otimes \complex^{d} \cong \Lin((\complex^{d})^*,\complex^{d})$, the variety $\sigma_r(\nu_2(\proj^{d-1}))$ corresponds to the set of (projective) symmetric $d \times d$ matrices of rank at most $r$.

For the Grassmannian variety $\Gr(m-1,\proj^{d-1}) \subseteq \proj({\bigwedge}^m(\complex^d))$, we have
\begin{align}
{\dim(\Gr(m-1,\proj^{d-1}))=m(d-m)},
\end{align}
and $\dim(\proj({\bigwedge}^m(\complex^d)))=\binom{d}{m}-1$, so
\begin{align}
\dim(\sigma_r(\Gr(m-1,\proj^{d-1})))\leq \min\left\{\binom{d}{m}-1, r m(d-m)+r-1\right\}.
\end{align}
As with the Segre variety, there are a conjecturally complete set of defective Grassmannian secants; see~\cite[Conjecture~7]{Bernardi_2018} and \cite{doi:10.1080/10586458.2007.10128997,boralevi2013note,bernardi2018new}. Similarly to the Segre and Veronese varieties, $\sigma_r(\Gr(1,\proj^{d-1}))$ is defective whenever $2 \leq r < \floor{\frac{d}{2}}$, in which case
\begin{align}
\dim(\sigma_r(\Gr(1,\proj^{d-1})))=\binom{d}{2}-\binom{d-2r}{2}-1=2r(d-r)-r-1.
\end{align}
Under the identification $\complex^{d}\otimes \complex^{d} \cong \Lin((\complex^{d})^*,\complex^{d})$, the variety $\sigma_r(\Gr(1,\proj^{d-1}))$ corresponds to the set of projective antisymmetric $d\times d$ matrices of rank at most $2r$.

\subsection{The set of projective $(n-1)$-planes as a projective variety}\label{sec:n-1plane}

In this subsection, we recall the canonical bijection between the Grassmannian variety and the set of projective planes of a fixed dimension, which endows the latter set with the structure of a projective variety. We then review several facts about the Grassmannian that we will use to prove our characterization of generic LUSD with a resource state (Theorem~\ref{USD_characterization}).

Recall the canonical bijection between the set of projective $(n-1)$-planes and ${\Gr(n-1,\proj \V)}$, which identifies $[v_1 \wedge \dots \wedge v_n]$ with $\spn\{[v_1],\dots, [v_n]\}$ for any linearly independent set $\{[v_1],\dots, [v_n]\}\subseteq \proj \V$. This identification endows the set of projective $n-1$ planes with the structure of a projective variety.
In a slight abuse of notation, we will use ${\Gr(n-1,\proj \V)}$ to refer to both the projective variety of decomposable elements of $\proj({\bigwedge}^n \V)$ and the projective variety of projective $(n-1)$-planes in $\proj \V$.

It is important to keep in mind the canonical bijection between the set of projective $(n-1)$-planes in $\proj \V$ (namely, $\Gr(n-1, \proj \V)$), and the set of $n$-dimensional linear subspaces of $\V$ (typically denoted $\Gr(n, \V)$). While $\Gr(n, \V)$ is perhaps more standard, we prefer $\Gr(n-1, \proj \V)$, as we would like the elements of each subspace to be states. Note that we have replaced the symbol $m$ with the symbol $n$ in this section, to match later notation in which the Grassmannian is viewed as a space of projective $(n-1)$-planes.

In accordance with our definition of \textit{generic}, we say that a property holds for a {generic} projective $(n-1)$-plane if there exists a Zariski-open-dense subset of $\setft{Gr}(n-1,\proj \V)$ on which the property holds. We say a property holds for a {generic} $n$-tuple of states in $\proj \V$ (or when $n=1$, simply a generic state in $\proj \V$) if there exists a Zariski-open-dense subset $U \subseteq \proj(V)^{\times n}$ such that the property holds for every $([v_1], \dots, [v_n]) \in U$. (Here, $\proj(V)^{\times n}$ is viewed as a projective variety via the Segre embedding.)

In the remainder of this subsection, we review two facts that we will use to prove our characterization of generic LUSD with a resource state. First, a generic $n$-tuple of states spans a generic projective $(n-1)$-plane, and vice versa (Fact~\ref{fact:generic}). Second, the bijection ${\Gr(n-1,\proj^d)\cong \Gr(d-n,\proj^d)}$, which sends a subspace to its orthogonal complement, defines an isomorphism of projective varieties (Fact~\ref{fact:ortho}). We defer the proofs of these facts to Appendix~\ref{app:proofs}.

\begin{fact}\label{fact:generic}
A generic $n$-tuple of states spans a generic projective $(n-1)$-plane, and vice versa. In more details, let $n$ be a positive integer, let $\V$ be a $\complex$-vector space, and let
\begin{align}\label{pi_tilde}
\tilde{\pi}: \proj(\V)^{\times n} \dashrightarrow \setft{Gr}(n-1,\proj \V)
\end{align}
be the rational map defined by $\tilde{\pi}([v_1],\dots,[v_n])=[v_1 \wedge \dots \wedge v_n]$. Then a subset ${U \subseteq \setft{Gr}(n-1,\proj \V)}$ is open-dense if and only if $\tilde{\pi}^{-1}(U) \subseteq \proj(\V)^{\times n}$ is open-dense.

\end{fact}
\begin{fact}\label{fact:ortho}
The bijection $\Gr(n-1,\proj^d)\cong \Gr(d-n,\proj^d)$, which sends a subspace to its orthogonal complement with respect to some non-degenerate bilinear form $\ip{\cdot}{\cdot}$, defines an isomorphism of projective varieties.
\end{fact}
\subsection{Projective linear subspaces disjoint from a variety}\label{sec:alg_geo_result}
The following algebraic-geometric result will be instrumental in proving our characterization of generic LUSD with a resource state. We will also use it to 
determine the maximum dimension of an entangled subspace.


\begin{theorem}\label{dimension_equivalence}
Let $X \subseteq \proj^D$ be a quasiprojective variety. Then
\begin{align}
\dim(X)=&\text{The smallest non-negative integer $d$ for which a generic projective $D-d-1$}\\
&\quad \text{plane is disjoint from $X$.}\\
=&\text{The largest non-negative integer $d$ for which a generic projective $D-d$}\\
&\quad \text{plane intersects $X$.}
\end{align}
If $X$ is projective, then
\begin{align}
\dim(X)=&\text{The largest non-negative integer $d$ for which every projective $D-d$ plane}\\
&\quad \text{intersects $X$.}
\end{align}
\end{theorem}

This characterization is taken to be the definition of the dimension of an irreducible variety in~\cite[Definition 11.2]{harris2013algebraic}, and is shown to be equivalent to other standard notions of dimension, e.g. the Krull dimension. The above extension to the reducible case is straightforward.
\subsection{The SLOCC image}\label{sec:LOCC}
\textit{LOCC channels} are quantum channels (completely positive, trace preserving maps) that can be implemented by local operations and classical communication (LOCC). \textit{SLOCC maps} are completely positive, trace non-increasing maps that can be implemented with non-zero probability by an LOCC channel. In other words, SLOCC maps represent LOCC channels with postselection \cite{watrous_2018}.

Let $\V=\bigotimes_{j=1}^m \complex^{d_j}$ and $\W=\bigotimes_{j=1}^{m} \complex^{c_j}$ be vector spaces. For a state $[w] \in \proj{\W}$, we define the \textit{SLOCC image} of $[w]$ in $\proj{\V}$, denoted $\Ima_{\proj{\V}}([w]) \subseteq \proj{\V}$, to be the set of states in $\proj{\V}$ obtainable from $[w]$ by SLOCC. This set is also known as the \textit{downward closure of $[w]$ with respect to SLOCC}, and can alternatively be characterized as the set of pure states that have a Tucker decomposition with core state $[w]$ \cite{Tucker:1966aa,rabanser2017introduction}. A related notion is the \textit{SLOCC orbit} of a state $[v]\in \proj \V$, denoted $\mathcal{O}_{[v]}\subseteq \proj{\V}$, which is the set of states in $\Ima_{\proj{\V}}([v])$ that can be converted back to $[v]$ by SLOCC, i.e.
\begin{align}
\mathcal{O}_{[v]} = \{[u] \in \Ima_{\proj{\V}}([v]) : [v] \in \Ima_{\proj{\V}}([u])\}.
\end{align}

In this subsection, we observe several properties of $\Ima_{\proj{\V}}([w])$ and $\mathcal{O}_{[v]}$ that we will use in Section~\ref{sec:genericLUSD}. We can describe these sets mathematically as
\begin{align}
\Ima_{\proj{\V}}([w])=\{[(A_1 \otimes \dots \otimes A_m)w] : A_i \in \Lin(\complex^{c_i},\complex^{d_i}) \quad \!\text{for all} &\quad \! i\in[m]\\
&\text{and}\quad \! (A_1 \otimes\dots \otimes A_m) w \neq 0\},
\end{align}
and
\begin{align}
\mathcal{O}_{[v]}=\{[(A_1 \otimes\dots \otimes A_m) v] : A_i \in \setft{GL}(\complex^{d_i})\quad\text{for all}\quad  i \in [m] \},
\end{align}
see \cite{PhysRevA.62.062314}. True to its name, $\mathcal{O}_{[v]}$ is the orbit of $[v]$ under the standard action of the product (projective) general linear group.

If $\V=\complex^{d_1}\otimes\complex^{d_2}$ and $\W=\complex^{c_1}\otimes \complex^{c_2}$ are bipartite spaces, then $\Ima_{\proj \V}([w])$ is the set of states in $\proj \V$ of Schmidt rank less than or equal to the Schmidt rank of $[w]$, and $\mathcal{O}_{[v]}$ is the set of states of Schmidt rank equal to the Schmidt rank of $[v]$. In multipartite space, consider the tensor-rank-$r$ GHZ state
\begin{align}\label{GHZ}
[\tau_{r,m}]=\left[\sum_{a=1}^r e_a^{\otimes m} \right] \subseteq \proj((\complex^r)^{\otimes m}).
\end{align}
It is straightforward to verify that $\Ima_{\proj{\V}}([\tau_{r,m}])$ is the set of states in $\proj{\V}$ of tensor rank at most $r$, and hence ${\overline{\Ima_{\proj{\V}}([\tau_{r,m}])}=\sigma_r}(Y)$, where $Y$ is the Segre variety of $\proj \V$ defined in~\eqref{seg}. It is clear that $\mathcal{O}_{[\tau_{r,m}]}$ is the set of states of the form $[\sum_{a=1}^r x_{a,1} \otimes \dots \otimes x_{a,m}]$, where $\{x_{1,j},\dots, x_{r,j}\}\subseteq \complex^{r}$ is linearly independent for all $j \in [m]$.

We conclude this subsection by proving that $\Ima_{\proj{\V}}([w])$ and $\mathcal{O}_{[v]}$ are both irreducible and constructible (Fact~\ref{fact:irred_construct}), and that in many cases a generic state $[w] \in \proj \W$ maximizes $\dim(\overline{\Ima_{\proj \V}([w])})$ (Fact~\ref{fact:dim}). We will use Facts~\ref{fact:irred_construct} and~\ref{fact:dim}, respectively, to prove our characterization of generic LUSD with a resource state, and to prove that a generic resource state is maximally useful for LUSD (see Theorem~\ref{USD_characterization} and the subsequent discussion). We defer the proofs of these facts to Appendix~\ref{app:proofs}.

\begin{fact}\label{fact:irred_construct}
Let $\V=\bigotimes_{j=1}^m \complex^{d_j}$ and $\W=\bigotimes_{j=1}^{m} \complex^{c_j}$ be vector spaces, and let $[w] \in \proj{\W}$ and ${[v] \in \proj \W}$ be states. Then the sets $\Ima_{\proj{\V}}([w])$ and $\mathcal{O}_{[v]}$ are both irreducible and constructible in the Zariski topology.
\end{fact}

Since both  $\Ima_{\proj{\V}}([w])$ and $\mathcal{O}_{[v]}$ are constructible, each contains an open-dense subset of its closure~\cite[Lemma 2.1]{An:2012aa}. In fact, $\mathcal{O}_{[v]}$ is itself an open-dense subset of its closure (i.e. $\mathcal{O}_{[v]}$ is \textit{locally closed})~\cite[Proposition 8.3]{humphreys2012linear}. Observe that $\overline{\mathcal{O}_{[v]}}=\overline{\Ima_{\proj{\V}}([v])}$.

\begin{fact}\label{fact:dim}
Let $\V=\bigotimes_{j=1}^m \complex^{d_j}$ and $\W=\bigotimes_{j=1}^{m} \complex^{c_j}$ be vector spaces with $c_j \leq d_j$ for all $j \in [m]$. Then $\dim(\overline{\Ima_{\proj \V}([w])})$ is maximized for a generic state $[w] \in \proj \W$.
\end{fact}
It would be nice to know if the condition that $c_j \leq d_j$ for all $j \in [m]$ is necessary for Fact~\ref{fact:dim} to hold.

\subsection{Local unambiguous state discrimination (LUSD)}\label{sec:LUSD}

We conclude this section by reviewing {unambiguous state discrimination} (USD), and its local counterpart, {local unambiguous state discrimination} (LUSD). We use these notions in Section~\ref{sec:genericLUSD} to characterize generic LUSD with a resource state.

An $n$-tuple of states $([v_1],\dots,[v_n]) \in \proj(\V)^{\times n}$ is \textit{(unambiguously) discriminable} if there exists a quantum measurement with $n+1$ outcomes $\{1,\dots,n,?\}$ that, when performed on any $[v_a]$, outputs either $a$ or $?$, with non-zero probability to output $a$. Mathematically, this is equivalent to the existence positive semidefinite operators $M_1,\dots,M_n, M_? \in \Pos(\V)$ for which $M_1+\dots+M_n+M_?=\I$ and $(\braket{v_b}{M_a v_b} \neq 0 \iff a = b)$.  Note that $([v_1],\dots,[v_n])$ is discriminable if and only if it is linearly independent.

Let $\V = \bigotimes_{j=1}^m \complex^{d_j}$ be a vector space. We say that an $n$-tuple of states
\begin{align}
{([v_1],\dots,[v_n]) \in \proj(\V)^{\times n}}
\end{align}
is \textit{locally (unambiguously) discriminable} if it is discriminable by a measurement implementable by an LOCC channel, with local subsystems $\proj \complex^{d_j}$. We say that an $n$-tuple of states $([v_1],\dots,[v_n])$ is \textit{locally (unambiguously) discriminable with (resource state)} ${[w] \in \proj{\W}}$ if $([v_1 \otimes w], \dots, [v_n \otimes w])$ is locally discriminable, where the local subsystem $j\in [m]$ is now the composite system $\proj (\complex^{d_j} \otimes \complex^{c_j})$. This is equivalent to $([v_1],\dots,[v_n])$ being locally discriminable via an LOCC measurement with pre-shared entanglement $[w]$ \cite{PhysRevA.94.022311}. 

\section{Generic local state discrimination with pre-shared entanglement}\label{sec:genericLUSD}

In this section, we characterize the maximum number of generic pure states that can be locally discriminated with a fixed resource state $[w]$, and observe that a generic resource state $[w]$ maximizes this number. To prove this characterization, we require the following mathematical description of LUSD with a resource state.

\begin{theorem}[\cite{PhysRevA.94.022311}]\label{cheflesext}
Let $\V=\bigotimes_{j=1}^m \complex^{d_j}$ and $\W=\bigotimes_{j=1}^m \complex^{c_j}$. An $n$-tuple of states $([v_1],\dots,[v_n]) \in \proj(\V)^{\times n}$ is locally (unambiguously) discriminable with resource state $[w]\in \proj{\W}$ if and only if there exist states
\begin{align}
[u_1],\dots, [u_n] \in \Ima_{\proj \V} ([w])
\end{align}
for which $(u_a^\t {v_b} \neq 0 \iff a= b)$, where the transpose can equivalently be taken with respect to any product basis of $\V$.
\end{theorem}
Note that this statement indeed does not depend on the choice of product basis over which the transpose is taken, since any two product bases are related by a product change of basis $A=A_1 \otimes \dots \otimes A_m \in \GL(\V)$, and
\begin{align}
[u_a]\in \Ima_{\proj \V} ([w]) \iff [A u_a] \in \Ima_{\proj \V} ([w]).
\end{align}
Also note that, if the Hermitian inner product is preferred, an alternative (equivalent) statement is that there exist $[u_1],\dots,[u_n] \in \Ima_{\proj \V} ([\overline{w}])$ such that ${({\braket{u_a}{v_b}}\neq 0 \iff a= b)}$, where $\overline{w}$ denotes the complex conjugate of $w$ with respect to any product basis of $\W$.

The SLOCC image of a trivial (i.e. non-existent) resource state is simply the Segre variety $Y$ of product states, defined in~\eqref{seg}. This case of Theorem~\ref{cheflesext}, first proven by Chefles, was used to characterize generic LUSD in~\cite{article}, and we use the above generalization of Bandyopadhyay et al. to characterize generic LUSD with a resource state.

\begin{theorem}\label{USD_characterization}
Let $n$ be a positive integer, let $\V=\bigotimes_{j=1}^m \complex^{d_j}$ and $\W=\bigotimes_{j=1}^m \complex^{c_j}$ be vector spaces, let $[w] \in \proj \W$ be a quantum state, let $X=\overline{\Ima_{\proj{\V}} ([w])}$ be the closure of the SLOCC image of $[w]$, and let $d= \dim(X)$. If $n \leq d+1$, then a generic $n$-tuple of states in $\proj \V$ is locally (unambiguously) discriminable with $[w]$. If $n> d+1$, then a generic $n$-tuple of states in $\proj \V$ is not locally (unambiguously) discriminable with $[w]$.
\end{theorem}
By Fact~\ref{fact:dim}, if $c_j \leq d_j$ for all $j \in [m]$, then $\dim(\overline{\Ima_{\proj{\V}} ([w])})$ is maximized for a generic resource state $[w] \in \proj \W$. It follows from Theorem~\ref{USD_characterization} that a generic resource state $[w]$ can be used to discriminate the maximum number of generic states in this setting. It would be nice to know if the condition that $c_j \leq d_j$ for all $j \in [m]$ can be dropped from Fact~\ref{fact:dim}, as this would imply that a generic resource state can be used to discriminate the maximum number of generic states under any choice of $\V$ and $\W$.

We emphasize that the third sentence of Theorem~\ref{USD_characterization} is much stronger than a simple converse to the second sentence: It asserts that an open-dense subset of $\proj (\V)^{\times n}$ is not locally discriminable with $[w]$.

We remark that similar characterizations of generic LUSD can be obtained in symmetric and antisymmetric space, but we omit stating such results explicitly.

\begin{example}
	Suppose $m = 3$ and $c_j = d_j = 2$ for all $j$, so that $\V = \W = \complex^2 \otimes \complex^2 \otimes \complex^2$ is three-qubit space. Then there are six different SLOCC orbits, with projective dimensions and representative members $[v] \in \mathcal{O}_{[v]}$ as follows (see \cite[Table~10.3.1]{landsberg2012tensors}, for example):
	\begin{align*}
		[v] & = [e_1 \otimes e_1 \otimes e_1] & \Longrightarrow && \dim(\mathcal{O}_{[v]}) & = 3, \\
		[v] & = [e_1 \otimes e_1 \otimes e_1 + e_1 \otimes e_1 \otimes e_2] & \Longrightarrow && \dim(\mathcal{O}_{[v]}) & = 4, \\
		[v] & = [e_1 \otimes e_1 \otimes e_1 + e_1 \otimes e_2 \otimes e_1] & \Longrightarrow && \dim(\mathcal{O}_{[v]}) & = 4, \\
		[v] & = [e_1 \otimes e_1 \otimes e_1 + e_2 \otimes e_1 \otimes e_1] & \Longrightarrow && \dim(\mathcal{O}_{[v]}) & = 4, \\
		[v] & = [e_1 \otimes e_1 \otimes e_2 + e_1 \otimes e_2 \otimes e_1 + e_2 \otimes e_1 \otimes e_1] & \Longrightarrow && \dim(\mathcal{O}_{[v]}) & = 6, \\
		[v] & = [\tau_{2,3}] = [e_1 \otimes e_1 \otimes e_1 + e_2 \otimes e_2 \otimes e_2] & \Longrightarrow && \dim(\mathcal{O}_{[v]}) & = 7.
	\end{align*}
	So, for example, Theorem~\ref{USD_characterization} tells us that a generic set of $4$ states is locally discriminable with a fully separable resource state, a generic set of $5$ states is locally discriminable with a bi-separable resource state, a generic set of $7$ states is locally discriminable with the ``W state'' $[e_1 \otimes e_1 \otimes e_2 + e_1 \otimes e_2 \otimes e_1 + e_2 \otimes e_1 \otimes e_1]$ as a resource, and a generic set of $8$ states is locally discriminable with the GHZ state $[\tau_{2,3}]$ as a resource.
\end{example}
	
\begin{proof}[Proof of Theorem~\ref{USD_characterization}]
We will make use of the map $\tilde{\pi}$ defined in~\eqref{pi_tilde}, as well as several facts observed in Section~\ref{mp}. Suppose first that $n > d+1$. By Theorem~\ref{dimension_equivalence}, there exists an open-dense subset $V \subseteq \Gr(d_1\cdots d_m-n-1, \proj \V)$ for which every element (i.e. subspace) in $V$ is disjoint from $X$. By Facts~\ref{fact:generic} and~\ref{fact:ortho}, $\tilde{\pi}^{-1}(V^{\perp})\subseteq \proj(\V)^{\times n-1}$ is open-dense, so
\begin{align}
U:=\proj(\V) \times \tilde{\pi}^{-1}(V^{\perp}) \subseteq \proj(\V)^{\times n}
\end{align}
is open-dense. For any $([v_1],\dots, [v_n]) \in U$, it holds that
\begin{align}
\spn\{[v_2],\dots, [v_n]\}^{\perp} \cap X = \{\},
\end{align}
by definition of $V$. By Theorem~\ref{cheflesext}, $([v_1],\dots, [v_n])$ is not locally (unambiguously) discriminable with resource state $[w]$.



Conversely, suppose $n \leq d+1$. By Theorem~\ref{cheflesext}, the desired result is equivalent to the existence of an open-dense subset of $\proj(\V)^{\times n}$ contained in $\bigcap_{a=1}^n S_a$, where
\begin{align*}
S_a= \{([v_1],\dots,[v_n]): \text{there exists} \quad [u] \in \Ima_{\proj \V}([w]) \quad\text{such that} \quad (u^\t v_b \neq 0 \iff a= b) \}
\end{align*}
for each $a \in [n]$. By Fact~\ref{fact:irred_construct}, $\Ima_{\proj \V}([w])$ is constructible. It follows that $S_a$ is constructible, so it contains an open-dense subset of its closure~\cite[Lemma 2.1]{An:2012aa}. To complete the proof, it suffices to show that $\overline{S_a}=\proj(\V)^{\times n}$ all $a \in [n]$. We do so by constructing a subset of $S_a$ that is dense in $\proj(\V)^{\times n}$.

We take $a=n$ to ease the notation (the other $a\in [n]$ follow by symmetry). Let ${U \subseteq {\Ima_{\proj \V}([w])}}$ be an open-dense subset of $X$. Then $U$ is an (irreducible) quasiprojective variety of dimension $d$. By Theorem~\ref{dimension_equivalence} and the inequality $n \leq d+1$, there exists an open-dense subset $W \subseteq \Gr(d_1\cdots d_m-n, \proj \V)$ for which every element (i.e. subspace) in $W$ intersects $U$. Therefore, the set $Z=\tilde{\pi}^{-1}(W^\perp)\subseteq \proj(\V)^{\times n-1}$ is open-dense, and for each $v:=([v_1],\dots,[v_{n-1}]) \in Z$ there exists $[u_v] \in U$ with $u_v^\t{v_b}=0$ for all $b\in [n-1]$. Let
\begin{align}
T_v = \{ [v_n] \in \proj \V : u_v^\t v_n\neq 0\}.
\end{align}
The set
\begin{align}
V_n=\bigcup_{v \in Z} \{v\} \times T_v
\end{align}
is clearly contained in $S_n$. To complete the proof, we show that $\overline{V_n}=\proj(\V)^{\times n}$. For any open-dense subset $S \subseteq \proj(\V)^{\times n}$, there exists $v \in Z$ for which the set $(\{v\} \times \proj(\V)) \cap S$ is open-dense inside $\{v\} \times \proj(\V)$. Since $T_v \subseteq \proj \V$ is open-dense, it follows that ${\{v\} \times T_v \subseteq \{v\} \times \proj \V}$ is open-dense, so
\begin{align}
(\{v\} \times T_v) \cap S\subseteq \{v\} \times \proj(\V)
\end{align}
is open-dense. Thus, $V_n \cap S \neq \{\}$. Since $V_n$ intersects every open-dense subset ${S \subseteq \proj (\V)^{\times n}}$, it follows that $\overline{V_n}=\proj(\V)^{\times n}$. This completes the proof.
\end{proof}

\begin{cor}\label{cor:GHZ}
Let $n$ and $r$ be positive integers, let $\V=\bigotimes_{j=1}^m \complex^{d_j}$ be a vector space, let $Y \subseteq \proj \V$ be the Segre variety defined in~\eqref{seg}, let $X=\sigma_r(Y)$ be the $r$-th secant variety, and let $d=\dim(X)$. If $n\leq d+1$, then a generic $n$-tuple of states in $\proj (\V)$ is locally (unambiguously) discriminable with the tensor-rank-$r$ GHZ state $[\tau_{r,m}],$ defined in~\eqref{GHZ}. If $n > d+1$, then a generic $n$-tuple of states in $\proj (\V)$ is not locally (unambiguously) discriminable with $[\tau_{r,m}]$.
\end{cor}
Note that the standard upper bound on $\dim(\sigma_r(Y))$, reviewed in~\eqref{eq:dim_segre_secant}, yields the upper bound mentioned in the introduction on the number of generic states locally discriminable with $[\tau_{r,m}]$. The dimension of $\sigma_r(Y)$ is known in several cases (see Section~\ref{sec:secant}), in particular, when $r=1$ or $m=2$. The next two corollaries follow from these known dimensions.
\begin{cor}\label{WS_generic}
Let $n$ be a positive integer, and let $\V=\bigotimes_{j=1}^m \complex^{d_j}$ be a vector space. If ${n \leq \sum_{j=1}^m (d_j-1)+1}$, then a generic $n$-tuple of states in $\proj (\V)$ is locally (unambiguously) discriminable. If $n > \sum_{j=1}^m (d_j-1)+1$, then a generic $n$-tuple of states in $\proj (\V)$ is not locally (unambiguously) discriminable.
\end{cor}
Corollary~\ref{WS_generic} follows from the fact that the SLOCC image of a trivial (i.e. non-existent) resource state is simply the Segre variety, which has dimension $\sum_{j=1}^m (d_j-1)$. Corollary~\ref{WS_generic} strengthens Theorem 4.3 in \cite{article} in two ways: First, the notion of ``generic" used by Walgate and Scott is measure-theoretic, which is weaker than our algebraic-geometric definition of ``generic." Second, when $n > \sum_{j=1}^m (d_j-1)+1$, Walgate and Scott simply prove that the subset of $\proj(\V)^{\times n}$ consisting of locally discriminable $n$-tuples of states is not full measure, which is much weaker than our result that a full-measure (and in fact, open-dense) subset of $\proj(\V)^{\times n}$ is not locally discriminable.

The known dimensions of secant varieties in bipartite space allow us to quantify exactly how many generic states are locally discriminable with an arbitrary resource state:
\begin{cor}\label{cor:bipartite}
Let $n$ and $r$ be positive integers, let $\V=\complex^{d_1}\otimes\complex^{d_2}$ and $\W=\complex^{c_1}\otimes\complex^{c_2}$ be vector spaces. If $n \leq d_1d_2-(d_1-\min\{d_1,r\})(d_2-\min\{d_2,r\})$, then a generic $n$-tuple of states in $\proj (\V)$ is locally (unambiguously) discriminable with any Schmidt-rank-$r$ resource state $[w]\in \proj \W$. If ${n > d_1d_2-(d_1-\min\{d_1,r\})(d_2-\min\{d_2,r\})}$, then a generic $n$-tuple of states in $\proj (\V)$ is not locally (unambiguously) discriminable with any Schmidt-rank-$r$ resource state $[w]\in \proj \W$.
\end{cor}

For example, if $d_1 = d_2$ is large and $r$ is small and fixed, then Corollary~\ref{cor:bipartite} says that a Schmidt-rank-$r$ resource state can be used to locally discriminate a generic tuple of linearly many ($d^2 - (d-r)^2 = 2dr - r^2$) pure states. At the other extreme, if the resource state is maximally entangled (so $r = d$) then it can be used to discriminate a generic tuple of quadratically many ($d^2 - (d-r)^2 = d^2$) pure states.

\section{Entangled subspaces of maximum dimension}\label{sec:entangled_subspaces}

For a vector space $\V=\bigotimes_{j=1}^m \complex^{d_j}$, we define a \textit{(standard)} $r$\textit{-entangled (projective linear) subspace} of $\proj \V$ to be a projective linear subspace disjoint from $\sigma_r(Y)$, where $Y$ is the Segre variety of product states~\eqref{seg}. In other words, an $r$-entangled subspace is one that does not contain any states of border rank at most $r$. Similarly, we define a \textit{symmetric $r$-entangled (projective linear) subspace} of $\proj({\bigvee}^m \complex^d)$ to be a projective linear subspace disjoint from $\sigma_r(\nu_m(\proj^{d-1}))$, and an \textit{antisymmetric $r$-entangled (projective linear) subspace} of $\proj({\bigwedge}^m \complex^d)$ to be a projective linear subspace disjoint from $\sigma_r(\Gr(m-1,\proj^{d-1}))$. These three types of $r$-entangled subspaces correspond to entangled subspaces in systems of distinguishable, bosonic, and fermionic particles, respectively~\cite{Grabowski_2012}.

In Appendix~\ref{app:dim} we use Theorem~\ref{dimension_equivalence} to determine the maximum dimensions of these entangled subspaces. In particular, Corollaries~\ref{entangled_dimension},~\ref{symmetric_dimension}, and~\ref{antisymmetric_dimension} establish the maximum possible dimension of standard $r$-entangled, symmetric $r$-entangled, and antisymmetric $r$-entangled subspaces, respectively, and show that a generic subspaces of these dimensions are $r$-entangled. Despite the abundance of entangled subspaces of maximum dimension, explicit constructions of them are only known for standard $r$-entangled subspaces in the $r = 1$ case \cite{BHAT_2006} and the $m = 2$ case \cite{Cubitt_2008}. In this section, we match these results for symmetric and antisymmetric $r$-entangled subspaces by providing explicit constructions in these two cases. We also construct maximal $2$-entangled subspaces of $\proj(\complex^{d_1}\otimes \complex^{d_2}\otimes \complex^2)$ for $d_1,d_2 \in \{2,3\}$, and of $\proj(\complex^2 \otimes \complex^2 \otimes \complex^2 \otimes \complex^2)$.

Many of these constructions (and indeed, many of the already known constructions of similar subspaces from \cite{Cubitt_2008,Heinosaari_2013,CDJ13}) are based on \emph{totally non-singular matrices}, which are matrices with the property that all of their minors (i.e., determinants of square submatrices) are non-zero. For example, every Vandermonde matrix
\begin{align}\label{eq:vandermonde}
	 \begin{bmatrix}
		1 & \alpha_1 & \alpha_1^2 & \cdots & \alpha_1^{d-1} \\
		1 & \alpha_2 & \alpha_2^2 & \cdots & \alpha_2^{d-1} \\
		\vdots & \vdots & \vdots & \ddots & \vdots \\
		1 & \alpha_d & \alpha_d^2 & \cdots & \alpha_d^{d-1}
	\end{bmatrix}
\end{align}
with $\alpha_i \neq \alpha_j \neq 0$ for all $i \neq j$ is totally non-singular (a slightly stronger property of these matrices, called total positivity, was proved in the case when $0 < \alpha_1 < \cdots < \alpha_d$ in \cite{Fal01}, but the same proof works for total non-singularity in general). In fact, total non-singularity is a generic phenomenon: the set of totally non-singular matrices is open-dense.

The following result, which was proved in \cite[Lemma~9]{Cubitt_2008}, provides the reason that totally non-singular matrices are of use to us.

\begin{lemma}\label{lem:tot_nonsing_cols}
	Let $M$ be an $n \times n$ totally non-singular matrix with $n \geq k$, and let $v \in \complex^n$ be a linear combination of $k$ of the columns of $M$. Then $v$ contains at least $n-k+1$ non-zero entries.
\end{lemma}

\subsection{Maximal symmetric $r$-entangled subspaces of bipartite space}\label{sec:explicit_symmetric_bipartite}

To construct a symmetric $r$-entangled subspace of $\proj(\complex^d \vee \complex^d)$ attaining the bound~\eqref{symmetric_bipartite_dim}, we first note that the isomorphism $\complex^d \otimes \complex^d \cong \Lin((\complex^d)^*,\complex^d)$ shows that it is equivalent to construct a projective linear subspace of symmetric $d\times d$ matrices of rank greater than $r$ of dimension
\[
	\binom{d-r+1}{2}-1.
\]
We construct such a subspace by placing columns of totally non-singular matrices along the super- and sub-diagonals of those symmetric matrices. More specifically, for each ${0 \leq i \leq d-r-1}$ and $1 \leq j \leq d-r-i$, let $M_j^i$ be the matrix that has the $j$-th column of some $(d-i) \times (d-i)$ totally non-singular matrix along its $i$-th super-diagonal (where for each fixed $i$, the same totally non-singular matrix is used for all $j$). For example, if $d = 6$, $r = 2$, $i = 1$, and we choose the totally non-singular matrix to be the Vandermonde matrix~\eqref{eq:vandermonde} with $\alpha_j = j$ for $1 \leq j \leq d-i = 5$, then
\[
	M_1^1 = \begin{bmatrix}
		0 & 1 & 0 & 0 & 0 & 0 \\
		0 & 0 & 1 & 0 & 0 & 0 \\
		0 & 0 & 0 & 1 & 0 & 0 \\
		0 & 0 & 0 & 0 & 1 & 0 \\
		0 & 0 & 0 & 0 & 0 & 1 \\
		0 & 0 & 0 & 0 & 0 & 0 \\
	\end{bmatrix}, \ M_2^1 = \begin{bmatrix}
		0 & 1 & 0 & 0 & 0 & 0 \\
		0 & 0 & 2 & 0 & 0 & 0 \\
		0 & 0 & 0 & 3 & 0 & 0 \\
		0 & 0 & 0 & 0 & 4 & 0 \\
		0 & 0 & 0 & 0 & 0 & 5 \\
		0 & 0 & 0 & 0 & 0 & 0 \\
	\end{bmatrix}, \ M_3^1 = \begin{bmatrix}
		0 & 1 & 0 & 0 & 0 & 0 \\
		0 & 0 & 4 & 0 & 0 & 0 \\
		0 & 0 & 0 & 9 & 0 & 0 \\
		0 & 0 & 0 & 0 & 16 & 0 \\
		0 & 0 & 0 & 0 & 0 & 25 \\
		0 & 0 & 0 & 0 & 0 & 0 \\
	\end{bmatrix}.
\]

We claim that the following set of symmetric matrices is a basis of a symmetric $r$-entangled subspace of $\proj(\complex^d \vee \complex^d)$ (once we convert the matrices back into states in the canonical way):
\begin{align*}
	B = \Big\{ \left[M_j^i + \big(M_j^i\big)^\t \right]: 0 \leq i \leq d-r-1, 1 \leq j \leq d-r-i \Big\}.
\end{align*}

The fact that the set~$B$ is linearly independent (and thus a basis of its span) follows immediately from Lemma~\ref{lem:tot_nonsing_cols}: every non-zero linear combination of those basis matrices has, for each $0 \leq i \leq d-r-1$, at least $(d-i) - (d-r-i) + 1 = r+1 \geq 1$ non-zero entries along its $i$-th super-diagonal, and thus does not equal the zero matrix.

In fact, this argument also shows why this subspace is $r$-entangled: every non-zero diagonal of a matrix $M \in \mathrm{span}(B)$ contains at least $r+1$ non-zero entries, so there is an $(r+1) \times (r+1)$ submatrix of $M$ that is upper triangular with non-zero diagonal entries (and is thus invertible), so $\mathrm{rank}(M) \geq r+1$. Since the rank of a matrix corresponds to the symmetric tensor rank in $\proj(\complex^d \vee \complex^d)$, the result follows.

All that remains is to count the number of vectors in~$B$:
\begin{align*}
	|B| = \sum_{i=0}^{d-r-1}\big(d-r-i\big) & = (d-r)^2 - \sum_{i=0}^{d-r-1}i \\
	& = (d-r)^2 - \frac{1}{2}(d-r)(d-r-1) = \binom{d-r+1}{2}.
\end{align*}
Since the projective dimension of the subspace is $|B|-1$, this completes the proof.

\subsection{Maximal antisymmetric $r$-entangled subspaces of bipartite space}

To construct an antisymmetric $r$-entangled subspace of $\proj(\complex^d \wedge \complex^d)$ attaining the bound~\eqref{eq:antisy_bipartite}, we note that the isomorphism $\complex^d \otimes \complex^d \cong \Lin((\complex^{d})^*,\complex^{d})$ shows that it is equivalent to construct a projective linear subspace of antisymmetric $d \times d$ matrices with rank greater than $2r$ (not $r$) of dimension
\[
	\binom{d-2r}{2}-1.
\]

The construction of this subspace is identical to the symmetric construction from Section~\ref{sec:explicit_symmetric_bipartite}, except we omit the $M_j^i$ matrices with non-zero entries on the main diagonal (i.e., the ones with $i = 0$), and we subtract in the lower-triangular portion of each matrix instead of adding. That is, a basis of this subspace is
\begin{align*}
	B = \Big\{ \left[M_j^i - \big(M_j^i\big)^\t\right] : 1 \leq i \leq d-2r-1, 1 \leq j \leq d-2r-i \Big\},
\end{align*}
so the dimension of this subspace is
\begin{align*}
	|B|-1 = \binom{d-2r+1}{2} - (d-2r) -1= \binom{d-2r}{2}-1,
\end{align*}
as desired.

\subsection{Maximal symmetric $1$-entangled subspaces of multipartite space}

By Equation~\eqref{symmetric_dim} and the subsequent discussion, the maximum dimension of a symmetric $1$-entangled subspace of $\proj({\bigvee}^m \complex^d)$ is
\begin{align}
\binom{m+d-1}{m}-d-1.
\end{align}
In this section, we construct such a subspace.

Consider the subspace spanned by the linearly independent set
\begin{align}
\{[e_{a_1} \vee \cdots \vee e_{a_m}] : (a_1,\dots,a_m)\in [d]^{\times n} \setminus \Delta_d^n\},
\end{align}
where $\Delta_d^n=\{(a,\dots, a) : a  \in [d]\}$.
This subspace clearly has the correct dimension, and if
\begin{align}
\sum_{a \in [d]^{\times n} \setminus \Delta_d^n} \alpha_a (e_{a_1} \vee \cdots \vee e_{a_m})= x^{\otimes m}
\end{align}
for some $x =\sum_{b=1}^d \beta_b e_b$, then for each $b \in [d]$ the coefficient of $e_b^{\otimes m}$ in the expansion of $x^{\otimes m}$ is zero, and hence $\beta_b=0$. It follows that $x=0$, a contradiction.


\subsection{Maximal antisymmetric $1$-entangled subspaces of multipartite space}

By Equation~\eqref{antisymmetric_dim} and the subsequent discussion, the maximum dimension of an antisymmetric 1-entangled subspace of $\proj (\bigwedge^m \complex^d)$ is
\begin{align}
\binom{d}{m}-m(d-m)-2,
\end{align}
whenever $d \geq m$ (otherwise, $\bigwedge^m \complex^d=0$). We construct a subspace that attains this bound in a somewhat similar manner to the non-positive partial transpose subspaces constructed in \cite{Johnston_2013,Johnston_2019}. Let
\begin{align}
J=\left\{\binom{m}{2}+m-1,\binom{m}{2}+m,\dots, dm-\binom{m}{2}-1, dm-\binom{m}{2}\right\}.
\end{align}
For each $s \in J$, let
\begin{align}
I_s=\left\{(a_1,\dots, a_m) \in [d]^m : 1\leq a_1<\dots< a_m\leq d \quad \text{and} \quad \sum_{j=1}^m a_j =s\right\}.
\end{align}
Let
\begin{align}
\proj \W=\{ [v] \in \textstyle{\proj ( {\bigwedge}^m \complex^d)} : \sum_{a\in I_s} v_a =0\},
\end{align}
where $v_a$ is the coefficient of $e_{a_1} \wedge \dots \wedge e_{a_m}$ in the expansion of $v$ with respect to the standard basis
\begin{align}
\{[e_{a_1}\wedge \dots \wedge e_{a_m}] : 1\leq a_1<\dots< a_m\leq d\}
\end{align}
of $\proj ( {\bigwedge}^m \complex^d)$. We first observe that $\proj \W$ has the correct dimension. Note that ${\abs{J}=m(d-m)+1}$, and one linear constraint is placed on $\W$ for each $s \in J$, so
\begin{align}
\dim(\proj \W)=\dim(\textstyle{\proj ( {\bigwedge}^m \complex^d)})-\abs{J}=\binom{d}{m}-m(d-m)-2,
\end{align}
as desired.

To complete the proof, we need only show that $\proj \W$ does not contain any state of the form $[x_1 \wedge \dots \wedge x_m]$. Proving this property is quite technical, so we begin by proving it in the special case $m=2$ as a warm-up.

Let $[v] \in \proj \W$ be arbitrary, and let
\begin{align}
t= \min \{s \in J : v_a \neq 0 \quad \text{for some} \quad a \in I_s\}.
\end{align}
Let $a,b \in I_t$ be any two multi-indices for which $a \neq b$ and $v_a, v_b \neq 0$. Since ${a_1+a_2=b_1+b_2=t}$, there exists a permutation $\sigma \in S_2$ for which
\begin{align}
b_{\sigma(1)} < a_{\sigma(1)}<a_{\sigma(2)}<b_{\sigma(2)}.
\end{align}

Under the inclusion $\wedge^2 (\complex^d)\subseteq \V_1 \otimes \V_2$, where $\V_1, \V_2 \cong \complex^d$, we can regard $v$ as an element of $\Lin(\V_2^*,\V_1)$.
Under this identification, consider the $4 \times 4$ submatrix of $v$ corresponding to the column index $\{{b_{\sigma(1)}}, {a_{\sigma(1)}},{a_{\sigma(2)}},{b_{\sigma(2)}}\}$ and row index $\{{b_{\sigma(2)}},{a_{\sigma(2)}},{a_{\sigma(1)}},{b_{\sigma(1)}}\}$. It is straightforward to verify that this matrix takes the form
\begin{align}
  \begin{blockarray}{*{4}{c} l}
    \begin{block}{*{4}{>{$\footnotesize}c<{$}} l}
      $b_{\sigma(2)}$ & $a_{\sigma(2)}$ & $a_{\sigma(1)}$ & $b_{\sigma(1)}$ & \\
    \end{block}
    \begin{block}{[*{4}{c}]>{$\footnotesize}l<{$}}
      \pm v_b&0&0& 0& $b_{\sigma(1)}$ \\
       *&\pm v_a& 0&0 & $a_{\sigma(1)}$ \\
       *& * & \pm v_a& 0 & $a_{\sigma(2)}$ \\
      * & * & * & \pm v_b& $b_{\sigma(2)}$\\
    \end{block}
  \end{blockarray},
\end{align}
where an asterisk $(*)$ denotes an entry we don't care about, and a $\pm$ denotes a sign we don't care about. It follows that the rank of $v$ under this identification is at least four. This proves that $[v] \notin \Gr(1, \proj^{d-1})$, since any such state has rank two under this identification. This completes the proof that $\proj \W$ is antisymmetric 1-entangled in the case $m=2$.

We now proceed to prove that $\proj \W$ is antisymmetric 1-entangled for arbitrary $m$. As before, let $[v] \in \proj \W$ be arbitrary, let
\begin{align}
t= \min \{s \in J : v_a \neq 0 \quad \text{for some} \quad a \in I_s\},
\end{align}
and let $a,b \in I_t$ be any two multi-indices for which $a \neq b$ and $v_a, v_b \neq 0$. Let $A=\{a_1,\dots, a_m\}$, $B=\{b_1,\dots, b_m\}$, and let $S=A \cap B$. Let $q=m-\abs{S}$, and let ${s_{q+1}<\dots<s_m \in [d]}$ be such that $\{s_{q+1},\dots, s_{m}\}=S$.

Suppose toward contradiction that $[v]=[x_1 \wedge \dots \wedge x_m]$ for some $x_1,\dots, x_m \in \complex^d$. Writing $x_{j,a} \in \complex$ for the $a$'th coordinate of $x_j$ in the standard basis, we may assume that for any $j\in [m]$ and $i \in \{q+1,\dots,m\}$, it holds that $x_{j,s_i}\neq 0 \iff i=j$. This can be observed by considering the element of $\pi^{-1}[v]$ that is in reduced row echelon form with respect to the re-ordered standard basis in which the elements $e_{s_{q+1}},\dots, e_{s_{q+m}}$ come first. In this nice form, $v_a \neq 0$ implies
\begin{align}
w_{A\setminus S}\neq 0,
\end{align}
where $w=x_1 \wedge \dots \wedge x_q$, so $[w] \in \Gr(q-1,\proj^{d-1})$. The notation $w_{A\setminus S}$ denotes the coefficient of $w$ with respect to the standard basis element $e_{c_1} \wedge \dots \wedge e_{c_q}$, where ${c_1<\dots<c_q \in [d]}$ are such that $\{c_1,\dots, c_q\}=A\setminus S$. Similarly, $v_b \neq 0$ implies $w_{B \setminus S} \neq 0$. Let $Q=(A \cup B)\setminus S$, and let $\tilde{c}_1<\dots<\tilde{c}_{2q} \in [d]$ be such that $\{\tilde{c}_1,\dots, \tilde{c}_{2q}\}=Q$. Regarding $w$ as an element of $\Lin(\bigotimes_{k=2}^q \V_k^*,\V_1)$, consider the submatrix of $w$ with column index $Q$ and row index $\{\tilde{c}_{\hat{1}},\dots,\tilde{c}_{\hat{2q}}\}\subseteq [d]^{\times q-1}$, where

\begin{align}
\tilde{c}_{\hat{j}}=
\begin{cases}
(A \setminus S)\setminus \{\tilde{c}_{j}\} ,& \tilde{c}_j \in A \setminus S\\
(B \setminus S)\setminus \{\tilde{c}_{j}\} ,& \tilde{c}_j \in B \setminus S
\end{cases}\quad \text{for each} \quad j \in [2q].
\end{align}
Similarly to before, we have implicitly identified the set $(A \setminus S)\setminus \{\tilde{c}_{j}\}$ with a strictly increasing $q-1$-tuple, and likewise for $B$. By the minimality of $t$, it follows that this submatrix is lower-triangular with each diagonal entry equal to $\pm w_{A \setminus S}$ or $\pm w_{B \setminus S}$, and hence has rank $2q$. This is a contradiction, as any element of $\Gr(q-1,\proj^{d-1})$ has rank $q$ when viewed in this way. This completes the proof that $\proj \W$ is antisymmetric 1-entangled.

\subsection{Some maximal $2$-entangled subspaces of multipartite space}\label{sec:multipartite_higher_rank}

There are a handful of higher-rank multipartite cases where it is straightforward to construct a $2$-entangled subspace of $\proj(\bigotimes_{j=1}^m\mathbb{C}^{d_j})$ of maximum dimension, and we consider these well known.

For example, if $m = 3$ and $d_1 = d_2 = d_3 = 2$ then the maximum dimension of a $2$-entangled subspace is zero, and an explicit example is simply given by
\begin{align}
[e_2 \otimes e_1 \otimes e_1+e_1\otimes e_2 \otimes e_1+e_1 \otimes e_1 \otimes e_2].
\end{align}
Similarly, if $m = 3$ and $d_1 = 3$, $d_2 = d_3 = 2$ then we see that $2$-entangled subspaces cannot be any larger than $1$-dimensional, since we can identify $\complex^3 \otimes \complex^2 \otimes \complex^2 \cong \complex^3 \otimes \complex^4$ in the natural way and then use the bipartite bound from \cite{Cubitt_2008}. Furthermore, we can construct a $2$-entangled subspace of this dimension just by using the fact that a state has border rank $\geq 3$ if and only if it has flattening rank $\geq 3$~\cite[Theorem 5.1]{landsberg2004ideals}, so the explicit construction of a $2$-entangled subspace of $\proj(\complex^3 \otimes \complex^4)$ of dimension $1$ from \cite{Cubitt_2008} also works in $\proj(\complex^3 \otimes \complex^2 \otimes \complex^2)$.

In general, however, explicit constructions of $r$-entangled subspaces in this multipartite higher-rank setting are rather ad-hoc. Our contribution here is to present explicit examples of $2$-entangled subspaces of $\proj(\complex^3 \otimes \complex^3 \otimes \complex^2)$ and of $\proj(\complex^2 \otimes \complex^2 \otimes \complex^2\otimes \complex^2)$ (i.e., the smallest non-trivial and previously unknown cases) that are $5$-dimensional, which is maximal by Equation~\eqref{eq:2entangleddim}.

In $\proj(\complex^3 \otimes \complex^3 \otimes \complex^2)$, we first let $\delta,\epsilon,\theta,\kappa \in \complex^4$ and then consider the following set of $6$ states:
\begin{align}
B=\big\{ & \left[ \big(e_1 \otimes e_1+e_2 \otimes e_2 + e_3 \otimes e_3\big)\otimes e_1 \right],\\
& \left[\big(e_1 \otimes e_1+e_2 \otimes e_2+e_3 \otimes e_3 \big)\otimes e_2 \right],\\
& \left[\delta_1 (e_1 \otimes e_2 \otimes e_1)+\delta_2 (e_2 \otimes e_1 \otimes e_1)+\delta_3 (e_3 \otimes e_1 \otimes e_2)+\delta_4 (e_1 \otimes e_3 \otimes e_2)\right],\\
&\left[\epsilon_1 (e_1 \otimes e_2 \otimes e_1)+\epsilon_2 (e_2 \otimes e_1 \otimes e_1)+\epsilon_3 (e_3 \otimes e_1\otimes e_2)+\epsilon_4 (e_1 \otimes e_3 \otimes e_2)\right],\\
&\left[\theta_1 (e_1 \otimes e_2 \otimes e_1)+\theta_2 (e_3 \otimes e_1 \otimes e_1)+\theta_3 (e_3\otimes e_2 \otimes e_2)+\theta_4 (e_2 \otimes e_3 \otimes e_2)\right],\\
&\left[\kappa_1 (e_1 \otimes e_2 \otimes e_1)+\kappa_2 (e_3 \otimes e_1 \otimes e_1)+\kappa_3 (e_3\otimes e_2 \otimes e_2)+\kappa_4 (e_2 \otimes e_3 \otimes e_2)\right]
\big\}\label{eq:qutrittritbit_basis}
\end{align}

It is straightforward to show that if $\{\delta, \epsilon,\theta,\kappa\}$ is linearly independent then so is the set~\eqref{eq:qutrittritbit_basis}, so it is a basis of its span. Under the choice $\delta = (0,1,1,1)$, $\epsilon = (1,1,2,0)$, ${\theta = (1,1,1,0)}$, and $\kappa = (0,2,1,1)$; we have verified using the Macaulay2 software package \cite{M2} that $\spn(B)$ is indeed 2-entangled (i.e. all of its members have border rank $\geq 3$). Alternatively, in Appendix~\ref{app:entangled} we provide an explicit proof that $\spn(B)$ is 2-entangled. The details of this proof are somewhat more complicated than in the other constructions that we have considered (we elaborate on this point at the end of the appendix). We have developed a heuristic algorithm in MATLAB which suggests that most choices of $\{\delta, \epsilon,\theta,\kappa\}$ produce a 2-entangled subspace. Our Macaulay2 and MATLAB codes are available on github.com \cite{entangled_subspaces_code}.

In $\proj(\complex^2 \otimes \complex^2 \otimes \complex^2\otimes \complex^2)$, we similarly let $\phi,\psi \in \complex^4$ and then consider the following set of $6$ states:
\begin{align}
\Big\{ & \left[e_1 \otimes e_1 \otimes e_1 \otimes e_2+e_1 \otimes e_2\otimes e_2 \otimes e_1+e_2 \otimes e_2 \otimes e_1 \otimes e_2\right],\\
&\left[e_1 \otimes e_1 \otimes e_2 \otimes e_1+e_2 \otimes e_1 \otimes e_1 \otimes e_2+e_2 \otimes e_2 \otimes e_2 \otimes e_1\right],\\
&\left[e_1 \otimes e_2 \otimes e_1 \otimes e_1+e_2 \otimes e_1  \otimes e_1 \otimes e_1+e_2\otimes e_1 \otimes e_2 \otimes e_2+e_2 \otimes e_2 \otimes e_1 \otimes e_2\right],\\
&\left[e_1\otimes e_1 \otimes e_2 \otimes e_1+e_1 \otimes e_2\otimes e_1 \otimes e_1+e_1 \otimes e_2 \otimes e_2\otimes e_2+e_2 \otimes e_1 \otimes e_2 \otimes e_2\right],\\
&\left[\phi_1 (e_1^{\otimes 4})+\phi_2(e_1\otimes e_1 \otimes e_2\otimes e_2)+\phi_3(e_2\otimes e_2 \otimes e_1\otimes e_1)+\phi_4(e_2^{\otimes 4})\right],\\
&\left[\psi_1 (e_1^{\otimes 4})+\psi_2(e_1\otimes e_1 \otimes e_2\otimes e_2)+\psi_3(e_2\otimes e_2 \otimes e_1\otimes e_1)+\psi_4(e_2^{\otimes 4})\right]
\Big\}\label{eq:fourqubit_basis}
\end{align}


It is straightforward to show that if $\{\phi,\psi\}$ is linearly independent then so is the set~\eqref{eq:fourqubit_basis}, so it is a basis of its span. Furthermore, if $\phi = (0,1,1,1)$ and $\psi = (1,2,1,0)$ then we have verified via Macaulay2 that its span is $2$-entangled. As before, our numerics suggest that most choices of $\{\phi, \psi\}$ produce a 2-entangled subspace \cite{entangled_subspaces_code}.

\section{Witnesses and multipartite Schmidt number}\label{sec:witnesses}

This is the only section in which we consider mixed states, so we introduce notation and related definitions for them here. For a vector space $\V$ (which we always take to be over $\complex$), let $\density{\V}\subseteq \Lin(\V)$ be the set of positive semidefinite operators on $\V$ with trace one, which we refer to as the set of \textit{density operators}. We identify the set of rank-one density operators with the set of pure quantum states under the natural bijection $vv^* \mapsto [v]$. We call a density operator of rank greater than one a~\textit{mixed state}. Let $\herm{\V}\subseteq \Lin(\V)$ be the set of Hermitian operators on $\V$. For a positive semidefinite operator ${P \in \pos{\V}}$, we define the \textit{support} of $P$ to be the image of $P$, denoted $\Im(P)$. For a subset $Z \subseteq \V$ we say that $\rho$ is \textit{supported on} $Z$ if $\Im(P) \subseteq {Z}$. We say that $P,Q \in \pos{\V}$ have \textit{orthogonal support} if $\braket{v}{u}=0$ for all $v \in \Im(P)$ and $u \in \Im(Q)$.

The notion of an operator that ``witnesses'' a particular property of a (mixed) quantum state is an important one in quantum information theory---it provides a way of demonstrating that property via a single measurement, without the need to have complete information about the state.


\begin{definition}\label{defn:schmidt_number_multi}
	Given a subset $Z \subseteq \V$, a \emph{not-${Z}$ witness} is a Hermitian matrix ${W \in \herm{\V}}$ for which
	\begin{enumerate}
		\item $\tr(W\rho) \geq 0$ for all density operators $\rho \in \density{\V}$ supported on $Z$, and

		\item there exists some density operator $\sigma \in \density{\V}$ such that $\tr(W\sigma) < 0$.
	\end{enumerate}
\end{definition}


For example, if $\hat{Y} \subseteq \complex^{d_1} \otimes \complex^{d_2}$ is (the affine cone over) the set of product states defined in~\eqref{seg}, then not-$\hat{Y}$ witnesses are called \emph{entanglement witnesses} (since any mixed state $\sigma$ for which $\tr(W\sigma) < 0$ is then guaranteed to be entangled). Slightly more generally, if ${\hat{Z} \subseteq \complex^{d_1} \otimes \complex^{d_2}}$ is (the affine cone over) the set of states of tensor rank (i.e., Schmidt rank) at most $r$, then the not-$\hat{Z}$ witnesses are called $r$-entanglement witnesses \cite{SBL01} (or \emph{$r$-block positive} \cite{SSZ09}), and any mixed state $\sigma$ for which $\tr(W\sigma) < 0$ is said to have \emph{Schmidt number} greater than $r$ \cite{TH00}.

Even more generally, if $\hat{X} \subseteq \bigotimes_{j=1}^m \complex^{d_j}$ is (the affine cone over) the set of states with tensor rank at most $r$, then we still refer to not-$\hat{X}$ witnesses as $r$-entanglement witnesses, and the witnessed (mixed) states with support not in $\hat{X}$ are said to have (multipartite) Schmidt number larger than $r$ \cite{CYT16}. However, a hiccup that occurs in this multipartite case that does not in the bipartite case is that, since the set of states with (multipartite) Schmidt number $\leq r$ is not closed, it is not true that every mixed state with Schmidt number $> r$ can be detected by some $r$-entanglement witness. Indeed, only the mixed states that are outside of the \emph{closure} of that set of states can be detected, so these $r$-entanglement witnesses are better thought of as witnesses for border rank, not tensor rank.

It is well known, at least in the bipartite case, that entangled subspaces can be used to construct entanglement witnesses with the maximum number of negative eigenvalues (see \cite{Johnston_2013,Johnston_2019}, for example). We now show that the same is true in the multipartite case, and even for not-$Z$ witnesses in general, as long as $Z$ is a Euclidean closed cone.

\begin{theorem}\label{thm:neg_evals_subspace_corr}
	Suppose $Z \subseteq \V$ is a Euclidean closed cone. The maximum number of negative eigenvalues that a not-$Z$ witness can have is equal to the maximum dimension of a linear subspace $\W \subseteq \V$ for which $\W \cap Z = \{0\}$.
\end{theorem}

\begin{proof}
	Throughout this proof, we let $n$ be the maximum dimension of a linear subspace of $\V$ that trivially intersects $Z$.

	Suppose $W\in \herm{\V}$ is a not-$Z$ witness with $s$ negative eigenvalues. To see that $s \leq n$, write $W = W_{+} - W_{-}$ where $W_{+},W_{-}$ are positive semidefinite with orthogonal support (i.e., they come from the spectral decomposition of $W$). Letting $\mathcal{W}=\Im(W_{-})$, we have $\dim(\W) = s$. Furthermore, for all $v \in \mathcal{W}\setminus \{0\}$, it holds that
	\begin{align*}
		\tr\big(W v v^* \big) = \tr\big(W_{+}v v^*- W_{-}v v^*\big) = -\tr(W_{-} v v^*) < 0,
	\end{align*}
	from which it follows that $\mathcal{W}$ has trivial intersection with $Z$, so $s \leq n$.

	Conversely, to see that there is a not-$Z$ witness with $n$ negative eigenvalues, let $P$ be the orthogonal projection onto some linear subspace $\mathcal{W} \subseteq \V$ of dimension $n$ that trivially intersects $Z$. Then $\mathrm{Tr}(Pvv^*) < \braket{v}{v}$ for all $v \in Z$, and the fact that $Z$ is a closed cone implies that there exists a real constant $0 < \epsilon < 1$ such that $\mathrm{Tr}(Pvv^*) \leq \epsilon \braket{v}{v}$ for all $v \in Z$. If we define the matrix $W = \I - \tfrac{1}{\epsilon}P$, then $W$ has exactly $\dim(\W) = n$ negative eigenvalues. Furthermore, $\mathrm{Tr}(Wvv^*) \geq 0$ for all $v \in Z$, so $W$ is a not-$Z$ witness.
\end{proof}

Combining Theorem~\ref{thm:neg_evals_subspace_corr} with Theorem~\ref{dimension_equivalence}, we get that for any projective variety ${X \subseteq \proj^{D}}$, the maximum number of negative eigenvalues of a not-$\hat{X}$ witness is $D-\dim(X)$. We close this section by applying this result to entanglement witnesses.

\begin{cor}\label{cor:neg_evals_subspace_rent}
	The maximum number of negative eigenvalues that an $r$-entanglement witness can have is exactly
	\[
		d_1\cdots d_m - \dim(\sigma_r(Y)) - 1,
	\]
	where $Y = \mathrm{Seg}(\mathbb{P}^{d_1-1} \times \cdots \times \mathbb{P}^{d_m-1})$, as usual.
\end{cor}

For example, if $P \in \Pos(\complex^3 \otimes \complex^3 \otimes \complex^2)$ is the orthogonal projection onto the $6$-dimensional subspace spanned by (the affine cone over) the set $B$ in~\eqref{eq:qutrittritbit_basis}, then there exists some scalar $\alpha > 1$ such that $W = I - \alpha P$ is a $2$-entanglement witness (here we have set $\alpha = 1/\epsilon$, where $\epsilon < 1$ is as in the proof of Theorem~\ref{thm:neg_evals_subspace_corr}), and furthermore there is no $2$-entanglement witness with more negative eigenvalues than this one (which has $6$). Importantly, this gives (at least in principle) a \emph{measurable} way of showing that a multipartite state has large Schmidt number: if measuring a state $\rho \in \Pos(\complex^3 \otimes \complex^3 \otimes \complex^2)$ produces an expectation value $\tr(W\rho)$ that is strictly negative, then $\rho$ must have Schmidt number at least $3$.

Unfortunately, finding an explicit value of $\alpha > 1$ that actually works to make $W$ a $2$-entanglement witness is a non-trivial task, which we have not been able to solve analytically. Numerics performed in MATLAB, however, strongly suggest that the optimal choice of $\alpha$ is approximately $1.0113$.

A similar construction with $P \in \Pos(\complex^2 \otimes \complex^2 \otimes \complex^2 \otimes \complex^2)$ being the orthogonal projection onto the $6$-dimensional subspace spanned by (the affine cone over) the set $B$ in~\eqref{eq:fourqubit_basis} gives a $2$-entanglement witness $W = I - \alpha P \in \Lin(\complex^2 \otimes \complex^2 \otimes \complex^2 \otimes \complex^2)$ with $6$ negative eigenvalues (for some value of $\alpha > 1$).

We can analogously define a symmetric (or antisymmetric) $r$-entanglement witness to be a Hermitian operator $W \in \herm{\bigvee^m\mathbb{C}^{d}}$ (or $W \in \herm{\bigwedge^m\mathbb{C}^{d}}$) such that ${\tr(W\rho) \geq 0}$ for all $\rho \in\density{\bigvee^m\mathbb{C}^{d}}$ (or $\rho \in \density{\bigwedge^m\mathbb{C}^{d}}$) with rank $\leq r$, and for which there exists some $\sigma \in\density{\bigvee^m\mathbb{C}^{d}}$ (or $\sigma \in \density{\bigwedge^m\mathbb{C}^{d}}$) such that $\tr(W\sigma) < 0$.
\begin{cor}\label{cor:neg_evals_subspace_sym}
	The maximum number of negative eigenvalues that a symmetric $r$-entanglement witness can have is exactly
	\[
		\binom{d-1+m}{m}-\dim(\sigma_r(\nu_m(\proj^{d-1})))-1.
	\]
\end{cor}

\begin{cor}\label{cor:neg_evals_subspace_asym}
	The maximum number of negative eigenvalues that an antisymmetric $r$-entanglement witness can have is exactly
	\[
		\binom{d-1+m}{m}-\dim(\sigma_r(\Gr(m-1,\proj^{d-1})))-1.
	\]
\end{cor}

For example, in the bipartite $m = 2$ case, we see that the maximum number of negative eigenvalues that symmetric and antisymmetric $r$-entanglement witnesses can have is
\[
	\binom{d-r+1}{2} \quad \text{and} \quad \binom{d-2r}{2},
\]
respectively.

\bibliographystyle{plainnat}
\bibliography{Lovitz_Johnston_entangled_subspaces}

\section*{Appendices}

\begin{appendix}

\section{Proofs of facts from Section~\ref{mp}}\label{app:proofs}
In this appendix we prove Facts~\ref{fact:generic},~\ref{fact:ortho},~\ref{fact:irred_construct}, and~\ref{fact:dim}, which were presented in Section~\ref{mp}.
\begin{namedtheorem}{Fact~\ref{fact:generic}}
A generic $n$-tuple of states spans a generic projective $(n-1)$-plane, and vice versa. In more details, let $n$ be a positive integer, let $\V$ be a $\complex$-vector space, and let
\begin{align}
\tilde{\pi}: \proj(\V)^{\times n} \dashrightarrow \setft{Gr}(n-1,\proj \V)
\end{align}
be the rational map defined by $\tilde{\pi}([v_1],\dots,[v_n])=[v_1 \wedge \dots \wedge v_n]$. Then a subset ${U \subseteq \setft{Gr}(n-1,\proj \V)}$ is open-dense if and only if $\tilde{\pi}^{-1}(U) \subseteq \proj(\V)^{\times n}$ is open-dense.
\end{namedtheorem}
\begin{proof}
Consider the map
\begin{align}\label{pi}
\pi : \V^{\times n} &\dashrightarrow \setft{Gr}(n-1,\proj \V),
\end{align}
defined on the open subset $U \subseteq \V^{\times n}$ of linearly independent $n$-tuples, and given by $\pi(v_1,\dots , v_n)=[v_1 \wedge \dots \wedge v_n]$. We require the following claim, which is standard. For completeness, we prove this claim at the end of the proof.
\begin{claim}\label{claim:quotient}
The map $\pi$ defines a quotient (in the sense of Section II.6.3 in \cite{borel2012linear}) of $U$ by the algebraic group $\GL(\complex^n)$ under the action $A \cdot (v_1, \dots , v_n)=(v_1,\dots,v_n) A^{-1}$ for all $A \in \GL(\complex^n)$ and $(v_1,\dots, v_n)\in U$, where the product on the right is matrix multiplication.
\end{claim}
It follows that the Zariski topology on $\setft{Gr}(n-1,\proj \V)$ is precisely the quotient topology of $U$ by the equivalence relation induced by the group action. In other words, a subset $V\subseteq \setft{Gr}(n-1,\proj \V)$ is open if and only if $\pi^{-1}(V)$ is open. Passing to $\proj(\V)^{\times n}$, it follows that a subset $V\subseteq \setft{Gr}(n-1,\proj \V)$ is open if and only if $\tilde{\pi}^{-1}(V)$ is open. This completes the proof, modulo proving the claim.
\begin{proof}[Proof of Claim~\ref{claim:quotient}]
\renewcommand\qedsymbol{$\triangle$}
To prove that $\pi$ is a quotient map, note that both $U$ and $\setft{Gr}(n-1,\proj \V)$ are irreducible and smooth, so it suffices to verify that $\pi$ is a surjective, open map by \cite[Lemma II.6.2]{borel2012linear}. Surjectivity is obvious. To prove that $\pi$ is open, it suffices to prove that it is flat by \cite[Exercise III.9.1]{hartshorne2013algebraic}, which in turn follows from the fact that $U$ and ${\setft{Gr}(n-1,\proj \V)}$ are smooth, and the fibers of $\pi$ are equidimensional \cite[Exercise III.10.9]{hartshorne2013algebraic}.
\end{proof}
\end{proof}

\begin{namedtheorem}{Fact~\ref{fact:ortho}}
The bijection $\Gr(n-1,\proj^d)\cong \Gr(d-n,\proj^d)$, which sends a subspace to its orthogonal complement with respect to some non-degenerate bilinear form $\ip{\cdot}{\cdot}$, defines an isomorphism of projective varieties.
\end{namedtheorem}
\begin{proof}
Recall the map $\pi$ defined in~\eqref{pi}. After change of basis, we may assume ${\ip{u}{v}=u^\t v}$ for all $u,v \in \complex^{d+1}$. For any element $[v]=[v_1 \wedge \dots \wedge v_n]\in \Gr(n-1,\proj^d)$, there exists an element of $\pi^{-1}([v])$ of the form $\left(\frac{\I_n}{A}\right)$ for some $A \in \Lin(\complex^{n},\complex^{d+1-n})$. It is straightforward to verify that $\left(\frac{A^\t}{-\I_{d+1-n}}\right)\in \pi^{-1}([v]^\perp)$, where $[v]^\perp$ denotes the orthogonal complement of the projective $(n-1)$-plane $[v]$. One can also verify that, up to sign, the $n\times n$ minor of $\left(\frac{\I_{n}}{A}\right)$ corresponding to a column index set $S \subseteq [d+1]$ of size $\abs{S}=n$ is precisely the ${(d+1-n) \times (d+1-n)}$ minor of the matrix $\left(\frac{A^\t}{-\I_{d+1-n}}\right)$ corresponding to the column index set $[d+1] \setminus S$. The result follows from the fact that these minors are exactly the coordinates of $[v] \in \proj(\bigwedge^n \complex^d)$ and $[v]^\perp \in \proj(\bigwedge^{d+1-n} \complex^d)$ in the Plücker embedding, respectively \cite{harris2013algebraic}.
\end{proof}

\begin{namedtheorem}{Fact~\ref{fact:irred_construct}}
Let $\V=\bigotimes_{j=1}^m \complex^{d_j}$ and $\W=\bigotimes_{j=1}^{m} \complex^{c_j}$ be vector spaces, and let $[w] \in \proj{\W}$ and ${[v] \in \proj \W}$ be states. Then the sets $\Ima_{\proj{\V}}([w])$ and $\mathcal{O}_{[v]}$ are both irreducible and constructible in the Zariski topology.
\end{namedtheorem}
\begin{proof}
Note that $\Ima_{\proj{\V}}([w])$ is the image of the irreducible quasiprojective variety
\begin{align*}
Z=\{[A_1 \otimes \dots \otimes A_m] : A_i \in \Lin(\complex^{c_i},\complex^{d_i}) \quad \!\text{for all} \quad \! i\in[m]\quad\text{and}\quad \! (A_1 \otimes\dots \otimes A_m) w \neq 0\}
\end{align*}
under the morphism $w: Z \rightarrow \proj{\V}$ that sends $(A_1 \otimes\dots \otimes A_m)$ to $[(A_1 \otimes\dots \otimes A_m) w]$. Since $Z$ is constructible, $\Ima_{\proj{\V}}([w])$ is constructible by Chevalley's theorem~\cite[Theorem 3.16]{harris2013algebraic}. 
Since $Z$ is irreducible, $\Ima_{\proj{\V}}([w])$ is irreducible in the subspace topology. Nearly identical arguments show that $\mathcal{O}_{[v]}$ is also constructible and irreducible.
\end{proof}

\begin{namedtheorem}{Fact~\ref{fact:dim}}
Let $\V=\bigotimes_{j=1}^m \complex^{d_j}$ and $\W=\bigotimes_{j=1}^{m} \complex^{c_j}$ be vector spaces with $c_j \leq d_j$ for all $j \in [m]$. Then $\dim(\overline{\Ima_{\proj \V}([w])})$ is maximized for a generic state $[w] \in \proj \W$.
\end{namedtheorem}
\begin{proof}
Let $\iota: \proj \W \rightarrow \proj \V$ be the canonical inclusion map, which acts on product tensors $[x_1 \otimes \dots \otimes x_m]$ by appending $d_j-c_j$ zeroes to $x_j$ for each $j \in [m]$. Then for any state $[w] \in \proj \W$, it holds that
\begin{align}
\overline{\Ima_{\proj \V}([w])}&=\overline{\Ima_{\proj \V}(\iota([w]))}\\
						&=\overline{\mathcal{O}_{\iota([w])}}\subseteq \proj \V.
\end{align}
It is well known that for any non-negative integer $k$, the set
\begin{align}
S_k:=\{[v] \in \proj \V : \dim(\overline{\mathcal{O}_{[v]}}) \leq k\}\subseteq \proj \V
\end{align}
is Zariski closed~\cite[Lemma~1.14]{Brion2010}. It follows that $S_k \cap \iota(\proj \W) \subseteq \proj \V$ is Zariski closed, so
\begin{align}
\iota^{-1}(S_k \cap \iota(\proj \W))= \{[w] \in \proj \W : \dim(\overline{\Ima_{\proj \V}([w])}) \leq k\}\subseteq \proj \W
\end{align}
is Zariski closed. It follows that the set of $[w] \in \proj \W$ that maximize $\dim(\overline{\Ima_{\proj \V}([w])})$ is open-dense in $\proj \W$. This completes the proof.
\end{proof}

\section{The maximum dimensions of entangled subspaces}\label{app:dim}

In this appendix, we write down the maximum dimensions of entangled subspaces, by invoking Theorem~\ref{dimension_equivalence} and the known dimensions of secant varieties reviewed in Section~\ref{sec:secant}. Corollary~\ref{entangled_dimension} extends results in~\cite{Parthasarathy:2004aa} and~\cite{Cubitt_2008}.

\begin{cor}\label{entangled_dimension}
Let $\V=\bigotimes_{j=1}^m \complex^{d_j}$, and let
\begin{align}
Y=\setft{Seg}(\proj^{d_1-1}\times \dots\times \proj^{d_m-1})\subseteq \proj(\V)
\end{align}
denote the Segre variety of product states. The maximum dimension of an $r$-entangled subspace of $\proj(\V)$ is
\begin{align}
d_1\cdots d_m-\dim(\sigma_r(Y))-2,
\end{align}
and a generic projective linear subspace of this dimension is $r$-entangled.
\end{cor}
As a result, there always exists an $r$-entangled subspace of dimension
\begin{align}
d_1\cdots d_m- r\sum_{j=1}^m (d_j-1)-r-1,
\end{align}
whenever this quantity is non-negative. Furthermore, this is often the maximum dimension of an $r$-entangled subspace, with a conjecturally complete set of exceptions \cite{Abo_2008,Bernardi_2018}. If $r=1$ then this is the maximum dimension, which gives Parthasarathy's result~\cite{Parthasarathy:2004aa}.

In the bipartite case $\V=\complex^{d_1} \otimes \complex^{d_2}$, combining Corollary~\ref{entangled_dimension} with Equation~\eqref{eq:bipartite_sr_dim} gives that the maximum dimension of an $r$-entangled subspace is
\begin{align}
(d_1-r)(d_2-r)-1,
\end{align}
whenever $r\leq \min\{d_1,d_2\}$ (this is \cite[Theorem 11]{Cubitt_2008}). Under the isomorphism ${\complex^{d_1} \otimes \complex^{d_2} \cong \Lin((\complex^{d_1})^*,\complex^{d_2})}$, this is the maximum dimension of a projective linear subspace of $d_2 \times d_1$ matrices of rank greater than $r$.

If $r=2$ and $m\geq 3$, then $\sigma_r(Y)$ has the expected dimension, so the largest dimension of a $2$-entangled subspace in this case is
\begin{align}\label{eq:2entangleddim}
	d_1 \cdots d_m-2\sum_{j=1}^m (d_j-1)-3.
\end{align}
We explicitly construct maximal $2$-entangled subspaces in Section~\ref{sec:entangled_subspaces}.
\begin{cor}\label{symmetric_dimension}
Let $m$ and $d$ be positive integers, and let
\begin{align}
\nu_m(\proj^{d-1})\subseteq \textstyle{\proj({\bigvee}^m \complex^d)}
\end{align}
denote the Veronese variety of unentangled states in the symmetric space. The maximum dimension of a symmetric $r$-entangled subspace of $\proj({\bigvee}^m \complex^d)$ is
\begin{align}
\binom{d-1+m}{m}-\dim(\sigma_r(\nu_m(\proj^{d-1})))-2
\end{align}
and a generic projective linear subspace of this dimension is symmetric $r$-entangled.
\end{cor}
As a result, there always exists a symmetric $r$-entangled subspace of dimension
\begin{align}\label{symmetric_dim}
\binom{d-1+m}{m}-rd-1,
\end{align}
whenever this quantity is non-negative. This is the maximum dimension of a symmetric $r$-entangled subspace in many cases, with a known set of exceptions \cite{alexander1995polynomial,Bernardi_2018}. If $r=1$, then~\eqref{symmetric_dim} is the maximum dimension. If $m=2$, then the maximum dimension is given by
\begin{align}\label{symmetric_bipartite_dim}
\binom{d-r+1}{2}-1.
\end{align}
We explicitly construct maximal symmetric $r$-entangled subspaces in these two cases in Section~\ref{sec:entangled_subspaces}.

\begin{cor}\label{antisymmetric_dimension}
Let $m$ and $d$ be positive integers with $m \geq d$, and let
\begin{align}
\Gr(m-1,\proj^{d-1}) \subseteq \textstyle{\proj({\bigwedge}^m \complex^d)}
\end{align}
denote the Grassmannian variety of unentangled states in the antisymmetric space. The maximum dimension of an antisymmetric $r$-entangled subspace of  $\proj({\bigwedge}^m(\complex^d))$ is
\begin{align}
\binom{d-1+m}{m}-\dim(\sigma_r(\Gr(m-1,\proj^{d-1})))-2,
\end{align}
and a generic projective linear subspace of this dimension is antisymmetric $r$-entangled.
\end{cor}
As a result, there always exists an antisymmetric $r$-entangled subspace of dimension
\begin{align}\label{antisymmetric_dim}
\binom{d-1+m}{m}-1-rm(d-m)-r,
\end{align}
whenever this quantity is non-negative. This is often the maximum dimension of an antisymmetric $r$-entangled subspace, with a conjecturally complete set of exceptions~\cite{doi:10.1080/10586458.2007.10128997,Bernardi_2018}. If $r=1$, then~\eqref{antisymmetric_dim} is the maximum dimension. If $m=2$, then the maximum dimension is given by
\begin{align}\label{eq:antisy_bipartite}
\binom{d-2r}{2}-1.
\end{align}
We explicitly construct maximal antisymmetric $r$-entangled subspaces in these two cases in Section~\ref{sec:entangled_subspaces}.

\section{A 2-Entangled Qutrit-Qutrit-Qubit Subspace}\label{app:entangled}

In this appendix, we prove that the span of the set $B$ from Equation~\eqref{eq:qutrittritbit_basis} is $2$-entangled. In order to show this, we prove that each member of that span has a flattening with rank~$3$. Indeed, since $\mathbb{C}^3 \otimes \mathbb{C}^3 \otimes \mathbb{C}^2$ is naturally isomorphic to the space of $3 \times 6$ matrices, we can think of this subspace as consisting of block matrices of the form
\[
	M = \left[\begin{array}{ccc|ccc}
		\lambda & \alpha_1 & \beta_1 & \gamma & 0 & \alpha_4 \\
		\alpha_2 & \lambda & 0 & 0 & \gamma & \beta_4 \\
		\beta_2 & 0 & \lambda & \alpha_3 & \beta_3 & \gamma
	\end{array}\right],
\]
where $\alpha_j = \delta_j + \epsilon_j$ and $\beta_j = \theta_j + \kappa_j$ for all $1 \leq j \leq 4$ (and $\{\delta, \epsilon,\theta,\kappa\}$ is as in Section~\ref{sec:multipartite_higher_rank}). The partial transpose of $M$ is
\[
	M^\Gamma = \left[\begin{array}{ccc|ccc}
		\lambda & \alpha_2 & \beta_2 & \gamma & 0 & \alpha_3 \\
		\alpha_1 & \lambda & 0 & 0 & \gamma & \beta_3 \\
		\beta_1 & 0 & \lambda & \alpha_4 & \beta_4 & \gamma
	\end{array}\right],
\]
which is another flattening of this same state, so our goal is to show that $\mathrm{rank}(M) \geq 3$ or $\mathrm{rank}(M^\Gamma) \geq 3$.

Importantly, because of how we chose $\{\delta, \epsilon,\theta,\kappa\}$, we know if that if $\alpha_i\neq 0$ for any $i \in [4]$, then $\alpha_i \neq 0$ for at least three $i \in [4]$, and similarly for the $\beta_j$'s. Indeed, we saw the desirability of this property in Section~\ref{sec:entangled_subspaces}, where we repeatedly used Lemma~\ref{lem:tot_nonsing_cols}.

To show that $\mathrm{rank}(M) \geq 3$ or $\mathrm{rank}(M^\Gamma) \geq 3$, we now split into several cases depending on which of $\lambda$, $\gamma$, $\alpha_i$ and $\beta_j$ equal $0$.

\begin{itemize}
	\item[Case 1(a):] $\lambda = 0$, $\alpha_1,\alpha_2,\alpha_3 \neq 0$.
	
	The submatrix of $M$ corresponding to its $1$st, $2$nd, and $4$th columns, up to permutation similarity, has the form
	\begin{align*}
	\begin{bmatrix}
	\alpha_2 & 0 & 0\\
	\beta_2 & \alpha_3 & 0\\
	0 & \gamma & \alpha_1\\
	\end{bmatrix}
	\end{align*}
	which clearly has rank $3$ since it is triangular with non-zero diagonal entries.
\end{itemize}

The above case contains the flavor of most of the cases that we will consider, so from now on we just list which columns of $M$ or $M^\Gamma$ give rise to a submatrix that is (up to permutation similarity) triangular with non-zero diagonal entries, and thus has rank~$3$. For example, for Case~1(a) we would just now just say ``$M(1,2,4)$''.

\begin{itemize}
	\item[Case 1:] $\lambda = 0$.
	\begin{itemize}
		\item[(a):] $\alpha_1,\alpha_2,\alpha_3 \neq 0$. $M(1,2,4)$.

		\item[(b):] $\alpha_1,\alpha_2,\alpha_4 \neq 0$. $M^\Gamma(1,2,4)$.

		\item[(c):] $\alpha_1,\alpha_3,\alpha_4 \neq 0$.
		\begin{itemize}
			\item[(i):] $\gamma = 0$. $M^\Gamma(1,4,6)$.

			\item[(ii):] $\gamma \neq 0$. $M(2,4,5)$.
		\end{itemize}

		\item[(d):] $\alpha_2,\alpha_3,\alpha_4 \neq 0$.
		\begin{itemize}
			\item[(i):] $\gamma = 0$. $M(1,4,6)$.
			\item[(ii):] $\gamma \neq 0$. $M^\Gamma(2,4,5)$.
		\end{itemize}

		\item[(e):] $\alpha_1 = \alpha_2 = \alpha_3 = \alpha_4 = 0$.
		\begin{itemize}
			\item[(i):] $\gamma = 0$.
			This case is identical to Case~1(a--d) via symmetry (just rotate $M$ by $180$ degrees, which does not change its rank).

			\item[(ii):] $\gamma \neq 0$, $\beta_1 \neq 0$. $M^\Gamma(1,4,5)$.

			\item[(iii):] $\gamma \neq 0$, $\beta_2,\beta_3,\beta_4 \neq 0$. $M(1,4,6)$.
		\end{itemize}
	\end{itemize}

	\item[Case 2:] $\lambda \neq 0$, $\alpha_1 = \alpha_2 = \alpha_3 = \alpha_4 = 0$.
	\begin{itemize}
		\item[(a):] $\beta_3 = \beta_4 = 0$ (and thus $\beta_1 = \beta_2 = 0$ too). $M(1,2,3)$.

		\item[(b):] $\beta_3 \neq 0$.
		\begin{itemize}
			\item[(i):] $\gamma = 0$. $M(1,2,5)$.

			\item[(ii):] $\gamma \neq 0$. $M(2,4,6)$.
		\end{itemize}

		\item[(c):] $\beta_4 \neq 0$.
		\begin{itemize}
			\item[(i):] $\gamma = 0$. $M^\Gamma(1,2,5)$.

			\item[(ii):] $\gamma \neq 0$. $M(2,4,6)$.
		\end{itemize}
	\end{itemize}

	\item[Case 3:] $\lambda \neq 0$ and $\alpha_i \neq 0$ for at least three $i\in [4]$.
	\begin{itemize}
		\item[(a):] $\beta_3 = \beta_4 = 0$.
		\begin{itemize}
			\item[(i):] $\alpha_3 \neq 0$. $M^\Gamma(2,3,6)$.

			\item[(ii):] $\alpha_4 \neq 0$. $M(2,3,6)$.
		\end{itemize}

		\item[(b):] $\beta_3 \neq 0$, $\beta_4 = 0$.
		\begin{itemize}
			\item[(i):] $\alpha_4 = 0$, $\gamma = 0$. $M(1,3,4)$.

			\item[(ii):] $\alpha_4 = 0$, $\gamma \neq 0$. $M^\Gamma(4,5,6)$.

			\item[(iii):] $\alpha_4 \neq 0$, $\gamma = 0$. $M(2,3,6)$.

			\item[(iv):] $\alpha_4 \neq 0$, $\gamma \neq 0$, $\alpha_2 \neq 0$. $M^\Gamma(2,3,5)$.

			\item[(v):] $\alpha_4 \neq 0$, $\gamma \neq 0$, $\alpha_2 = 0$. This case is much more difficult, so we leave it until after the remaining cases are dealt with.
		\end{itemize}

		\item[(c):] $\beta_3 = 0$, $\beta_4 \neq 0$. This case is identical to Case~3(b) by taking the partial transpose (i.e., replace $M$ with $M^\Gamma$ and vice-versa).

		\item[(d):] $\beta_3, \beta_4 \neq 0$.
		\begin{itemize}
			\item[(i):] $\alpha_4 = 0$, $\gamma = 0$. $M(1,4,6)$.

			\item[(ii):] $\alpha_4 = 0$, $\gamma \neq 0$. $M^\Gamma(3,4,5)$.

			\item[(iii):] $\alpha_4 \neq 0$, $\beta_2 = 0$, $\gamma \neq 0$. $M^\Gamma(3,4,5)$.

			\item[(iv):] $\alpha_4 \neq 0$, $\beta_2 = 0$, $\gamma = 0$, $\alpha_3 = 0$. $M^\Gamma(1,3,6)$.

			\item[(v):] $\alpha_4 \neq 0$, $\beta_2 = 0$, $\gamma = 0$, $\alpha_3 \neq 0$. $M(2,3,4)$.

			\item[(vi):] $\alpha_4 \neq 0$, $\beta_2 \neq 0$, $\gamma = 0$. $M^\Gamma(2,3,4)$.

			\item[(vii):] $\alpha_4 \neq 0$, $\beta_2 \neq 0$, $\gamma \neq 0$. This is another difficult case that we deal with separately.
		\end{itemize}
	\end{itemize}
\end{itemize}

The only two remaining cases from above are 3(b)(v) and 3(d)(vii). These cases require a more intricate argument to demonstrate that at least one of $M$ or $M^\Gamma$ has rank~$3$, which we now provide.

\subsection*{Case 3(b)(v)}\label{sec:case3bv}

For this case, $\gamma,\lambda \neq 0$, $\alpha_2 = \beta_4 = 0$, $\alpha_1,\alpha_3,\alpha_4 \neq 0$, and $\beta_1,\beta_2,\beta_3 \neq 0$. Let ${\alpha=(\alpha_1,\alpha_2,\alpha_3, \alpha_4)}$ and $\beta=(\beta_1,\beta_2,\beta_3,\beta_4)$. Because of how we chose $\{\delta, \epsilon,\theta,\kappa\}$, we have
$\alpha = (\alpha_1,\alpha_1+\alpha_4,2\alpha_1+\alpha_4,\alpha_4)$ and $\beta =  (\beta_1,\beta_1+2\beta_4,\beta_1+\beta_4,\beta_4)$. Furthermore, using the facts that $0 = \alpha_2 = \alpha_1+\alpha_4$ and $0 = \beta_4$ shows that $\alpha = (\alpha_1,0,\alpha_1,-\alpha_1)$ and $\beta = (\beta_1,\beta_1,\beta_1,0)$. The matrix $M^\Gamma$ thus has the form
\begin{align*}
M^\Gamma & = \left[\begin{array}{ccc|ccc}
\lambda & 0 & \beta_1 & \gamma & 0 & \alpha_1 \\
\alpha_1 & \lambda & 0 & 0 & \gamma & \beta_1 \\
\beta_1 & 0 & \lambda & -\alpha_1 & 0 & \gamma
\end{array}\right].
\end{align*}
We can then see that $\mathrm{rank}(M^\Gamma) = 3$ as follows. If it had rank $\leq 2$ then it would be the case that $\det(M^\Gamma(1,2,3)) = \lambda(\lambda^2 - \beta_1^2) = 0$, which implies $\beta_1 = \pm \lambda$. We would similarly have $\det(M^\Gamma(4,5,6)) = \gamma(\gamma^2 + \alpha_1^2) = 0$, which implies $\alpha_1 = \pm i \gamma$. Finally, we would also have $\det(M^\Gamma(3,4,6)) = -\beta_1(\gamma\lambda + \alpha_1\beta_1) = 0$, which (since $\alpha_1 = \pm i \gamma$ and $\beta_1 = \pm \lambda$) implies $\gamma\lambda \pm i \gamma\lambda = 0$, which is impossible (recall that all of these variables are non-zero). It follows that at least one of these determinants is non-zero, so $\mathrm{rank}(M^\Gamma) = 3$, which completes this case.

\subsection*{Case 3(d)(vii)}

For this case, $\gamma,\lambda \neq 0$, $\alpha_4 \neq 0$, and $\beta_2,\beta_3,\beta_4 \neq 0$. We now split into four subcases:

\begin{itemize}
	\item $\alpha_1 = \beta_3 = 0$. This case is identical to Case~3(b)(v) from Section~\ref{sec:case3bv} by taking the partial transpose, swapping the roles of $\gamma$ and $\lambda$, and swapping the roles of $\alpha$ and $\beta$.

	\item $\alpha_1 = 0$, $\beta_3 \neq 0$. $M(2,4,5)$.

	\item $\alpha_1 \neq 0$, $\beta_3 = 0$. $M(2,3,5)$.

	\item $\alpha_1,\beta_3 \neq 0$. This is the difficult subcase. Similarly to Case 3(b)(v),
	we have ${\alpha= (\alpha_1,\alpha_1+\alpha_4,2\alpha_1+\alpha_4,\alpha_4)}$ and $\beta= (\beta_1,\beta_1+2\beta_4,\beta_1+\beta_4,\beta_4)$, so $M$ and $M^\Gamma$ have the form
	\begin{align*}
	M & = \left[\begin{array}{ccc|ccc}
	\lambda & \alpha_1 & \beta_1 & \gamma & 0 & \alpha_4 \\
	\alpha_1+\alpha_4 & \lambda & 0 & 0 & \gamma & \beta_4 \\
	\beta_1+2\beta_4 & 0 & \lambda & 2\alpha_1+\alpha_4 & \beta_1+\beta_4 & \gamma
	\end{array}\right] \quad \text{and} \\
	M^\Gamma & = \left[\begin{array}{ccc|ccc}
	\lambda & \alpha_1+\alpha_4 & \beta_1+2\beta_4 & \gamma & 0 & 2\alpha_1+\alpha_4 \\
	\alpha_1 & \lambda & 0 & 0 & \gamma & \beta_1+\beta_4 \\
	\beta_1 & 0 & \lambda & \alpha_4 & \beta_4 & \gamma
	\end{array}\right].
	\end{align*}
	Now suppose that $\rank(M) = \rank(M^\Gamma) = 0$. Then we would have
	\begin{align}
	{\det(M(2,3,6)) = \lambda(\lambda\alpha_4 - \alpha_1\beta_4 - \gamma\beta_1) = 0},
	\end{align}
	which implies $\lambda\alpha_4 - \gamma\beta_1 = \alpha_1\beta_4$. We would also have
	\begin{align}
	\det(M^\Gamma(1,4,5)) = \gamma(\gamma\beta_1 - \lambda\alpha_4 - \alpha_1\beta_4) = 0,
	\end{align}
	which implies $\lambda\alpha_4 - \gamma\beta_1 = - \alpha_1\beta_4$. Combining these two expressions for $\lambda\alpha_4 - \gamma\beta_1$ shows that $\alpha_1\beta_4 = 0$, which contradicts the fact that we are assuming that $\alpha_1,\beta_4 \neq 0$ in this case. It follows that at least one of these determinants is non-zero, which completes this final subcase and the proof.
\end{itemize}

Numerics suggest that almost all choices of $\{\delta, \epsilon,\theta,\kappa\}$ lead to this subspace being $2$-entangled. Indeed, the primary property of $\{\delta, \epsilon,\theta,\kappa\}$ that we made use of was that a non-zero linear combination of $\delta$ and $\epsilon$ must never have more than one $0$ entry, and similarly for $\theta$ and $\kappa$ (this property is generic). The only other place where the particular entries of these vectors was used was in Case~3(b)(v) (and the symmetric first subcase of Case~3(d)(vii)), where non-invertibility of a $3 \times 3$ submatrix of $M$ or $M^\Gamma$ was more delicate.

To illustrate why Case~3(b)(v) is more delicate, notice that if we had instead chosen $\delta = (1,1,1,2)$, $\epsilon = (0,-1,1,1)$, $\theta = (2,1,1,0)$, and $\kappa = (1,1,0,1)$, then it would still be the case that any non-zero linear combination of $\delta$ and $\epsilon$ would never have more than one $0$ entry, and similarly for $\theta$ and $\kappa$ (so all of the other cases still work fine). However, the proof would fall apart in Case~3(b)(v), since we would get the pair of matrices
\begin{align*}
M = \left[\begin{array}{ccc|ccc}
1 & -1 & 1 & 1 & 0 & 1 \\
0 & 1 & 0 & 0 & 1 & 0 \\
1 & 0 & 1 & 1 & 1 & 1
\end{array}\right] \quad \text{and} \quad M^\Gamma = \left[\begin{array}{ccc|ccc}
1 & 0 & 1 & 1 & 0 & 1 \\
-1 & 1 & 0 & 0 & 1 & 1 \\
1 & 0 & 1 & 1 & 0 & 1
\end{array}\right],
\end{align*}
both of which have rank $2$.

\end{appendix}
\end{document}